\documentclass[11pt,a4paper]{article}

\usepackage{cmap} 
\usepackage[T1]{fontenc}
\usepackage{lmodern} 

\usepackage[utf8]{inputenc}
\usepackage{amsmath,amssymb,amsthm,mathtools}
\usepackage[margin=2.5cm]{geometry}
\usepackage{enumitem}
\usepackage[colorlinks=true,linkcolor=black,citecolor=black,urlcolor=blue]{hyperref}
\usepackage{cleveref}
\usepackage{natbib}
\usepackage{booktabs}
\usepackage{float}
\usepackage{multirow}
\usepackage{algorithm}
\usepackage{algpseudocode}
\usepackage{tikz}
\usetikzlibrary{arrows.meta,positioning,shapes.geometric,fit,calc}
\usepackage{xcolor}
\usepackage{colortbl}
\usepackage{orcidlink}

\theoremstyle{plain}
\newtheorem{theorem}{Theorem}[section]
\newtheorem{proposition}[theorem]{Proposition}
\newtheorem{lemma}[theorem]{Lemma}

\theoremstyle{definition}
\newtheorem{definition}[theorem]{Definition}
\newtheorem{assumption}{Assumption}
\theoremstyle{remark}
\newtheorem{remark}[theorem]{Remark}

\newcommand{\dd}{\mathrm{d}}
\newcommand{\RR}{\mathbb{R}}
\newcommand{\EE}{\mathbb{E}}
\newcommand{\PP}{\mathbb{P}}

\usepackage{lineno}

\title{\textbf{Differentiable Stochastic Traffic Dynamics: Physics-Informed Generative Modelling in Transportation}}
\author{Wuping Xin \orcidlink{0000-0002-9021-5763}\\[4pt]\textit{Caliper Corporation}}
\date{March 2026\\[8pt]
{\small \copyright~2026 Wuping Xin. Licensed under CC BY-NC-ND 4.0.}}

\begin{document}
\maketitle

\begin{abstract}
Macroscopic traffic flow is stochastic, but the physics-informed deep learning methods currently used in transportation literature embed deterministic PDEs and produce point-valued outputs; the stochasticity of the governing dynamics plays no role in the learned representation. This work develops a framework in which the physics constraint itself is distributional and directly derived from stochastic traffic-flow dynamics. Starting from an Ito-type Lighthill-Whitham-Richards model with Brownian forcing, we derive a one-point forward equation for the marginal traffic density at each spatial location. The spatial coupling induced by the conservation law appears as an explicit conditional drift term, which makes the closure requirement transparent. Based on this formulation, we derive an equivalent deterministic Probability Flow ODE that is pointwise evaluable and differentiable once a closure is specified. Incorporating this as a physics constraint, we then propose a score network with an advection-closure module, trainable by denoising score matching together with a Fokker-Planck residual loss. The resulting model targets a data-conditioned density distribution, from which point estimates, credible intervals, and congestion-risk measures can be computed. The framework provides a basis for distributional traffic-state estimation and for stochastic fundamental--diagram analysis in a physics-informed generative setting.\footnotemark

\end{abstract}
\footnotetext{The present version focuses on theoretical derivation and methodological formulation; empirical validation is deferred to a subsequent revision.}

\newpage

\section{Introduction}\label{sec:intro}

\subsection{The disconnect between stochastic traffic physics and deep learning}

Macroscopic traffic flow has been understood as a stochastic phenomenon
for decades.  Day-to-day variability in driver behaviour, vehicle
composition, weather, and minor incidents produces persistent
randomness in the density--speed--flow relationship
\citep{sumalee2011stochastic, jabari2012stochastic, li2012analysis}.
Stochastic extensions of the Lighthill--Whitham--Richards (LWR) model
\citep{lighthill1955kinematic, richards1956shock} have been developed
to capture this variability, including random-parameter formulations
\citep{fan2023dynamically, fan2024stochastic}, stochastic cell
transmission models \citep{sumalee2011stochastic}, and models based on
random time headways \citep{jabari2012stochastic}.

In parallel, deep learning has transformed traffic analysis.
Physics-informed neural networks (PINNs) embed deterministic traffic
models into neural network training losses and have shown strong
data efficiency in traffic state estimation
\citep{shi2021physics, mo2021physics, huang2022physics}.
Score-based generative models \citep{song2019generative, song2021score}
have demonstrated the ability to learn and sample from complex
probability distributions across domains from image synthesis to
molecular design.

These two lines of work have largely developed separately.
Stochastic traffic models are formulated as stochastic PDEs that
require sampling-based solvers (e.g., Monte Carlo simulation with
thousands of realisations), while neural network training requires
differentiable operations that can be backpropagated through.  No
existing stochastic traffic model produces a physics constraint that
is compatible with the automatic differentiation pipelines of modern
deep learning.  Recent work has applied generic diffusion and
score-based generative models to traffic state estimation, but these
approaches borrow the probability-flow machinery from image synthesis
and treat it as a black-box generative tool; the traffic physics
plays no structural role in the diffusion process itself.  As a
consequence, the learned distributions are not constrained by
conservation-law dynamics, and the link between the generative
model and the underlying stochastic traffic physics remains loose.

Here we begin with a stochastic traffic-flow model and derive the corresponding distributional dynamics directly. Rather than importing a generic diffusion model and adding traffic constraints afterward, we derive the learning framework from the stochastic traffic model itself, introducing closure and neural components only where they are needed for computation.

\subsection{The core obstacle and our resolution}

The difficulty in combining stochastic traffic models with
deep learning is not only computational but also mathematical.
Consider the standard approach: one writes down a stochastic PDE for
traffic density, and one would like to derive an equation governing how
the \emph{probability distribution} of density evolves over time.
If such an equation existed and were deterministic, it could be
evaluated pointwise and differentiated via automatic differentiation,
making it directly usable as a physics constraint in a neural network.

The obstacle is that for conservation-law--based traffic models like
LWR, this distributional equation is not self-contained.  The local
density evolution at any point $x$ depends on the spatial gradient
$\partial_x\rho$, which couples the probability distribution at $x$ to
the distributions at neighbouring locations.  This spatial coupling has
historically prevented the derivation of a closed, tractable evolution
equation for the one-point density distribution, and is the fundamental
reason why stochastic LWR dynamics have not been connected to
distributional deep learning.

We address this obstacle through three steps.  First, we formulate an
It\^{o}-type stochastic LWR model with dynamic Brownian forcing and
derive the exact one-point Fokker--Planck equation under a
smooth-solution regime, writing the spatially coupled drift as an
explicit conditional expectation rather than hiding it in notation.
Second, we make the required closure transparent and provide tractable
options, enabling the construction of an equivalent deterministic
Probability Flow ODE.  Third, we show that, once the closure is
specified, this Probability Flow ODE is pointwise evaluable and fully
compatible with automatic differentiation, so it can serve as a
direct, computable distributional physics constraint within a
score-based neural network.

Within the smooth-solution regime, the resulting one-point forward representation of stochastic conservation-law traffic dynamics becomes compatible with gradient-based deep learning once a closure is specified. This closes the gap between stochastic traffic modelling and distributional learning in a form that can be used directly in training.

\subsection{Significance beyond a single application}

The implications of deriving the generative structure from the traffic
physics, rather than borrowing it from a generic framework, extend
well beyond any single application.  We highlight three.

\paragraph{Distributional traffic state estimation.}
Existing methods, from extended Kalman filters
\citep{wang2008real, wang2009adaptive} to PIDL
\citep{shi2021physics}, produce point estimates of traffic density.
Uncertainty quantification, when attempted, is post-hoc and not
grounded in the governing physics.  Our framework targets a
distributional estimate
$p_\theta(\hat\rho; x, t)$, trained on the available observations
$\mathcal{O}$, from
which point estimates, credible intervals, and congestion-risk measures
follow directly, with the estimate itself
trained to approximately satisfy the one-point stochastic traffic
physics.

\paragraph{Stochastic fundamental diagram.}
The well-documented scatter in empirical flow--density measurements,
both at the link level and in the network-level macroscopic
fundamental diagram (MFD) \citep{geroliminis2008existence,
mazloumian2010spatial}, has lacked a principled stochastic
characterisation.  Because our framework produces the full density
distribution at each location, it naturally introduces a link-level stochastic fundamental diagram.  At the network level,
spatial aggregation of the distributional dynamics suggests a corresponding MFD, in which part of the observed scatter may be interpreted as a consequence of the underlying Brownian forcing rather
than unexplained noise.

\paragraph{Physics-informed generative modelling for traffic.}
More broadly, by showing that a stochastic conservation law can
be converted into a differentiable probability flow, this work
suggests an architectural pattern for physics-informed generative
models in transportation.  The same mathematical machinery,
comprising an exact forward equation, a conditional drift closure,
and a probability flow ODE, can in principle be extended to
second-order models (e.g., ARZ), multi-class traffic, and other
conservation-law systems in transportation science.  This
work shows how stochastic traffic-flow theory can be formulated
as a trainable distributional estimation framework.

\subsection{Contributions}

The contributions are ordered so that each one supports the next.
The generative structure in this framework is not borrowed from a generic diffusion model; it is
derived from the stochastic traffic physics.

\begin{enumerate}
\item \textbf{Stochastic LWR model with dynamic Brownian forcing}
  (\Cref{sec:model}).  We formulate an It\^{o}-type stochastic LWR
  equation driven by finite-dimensional Brownian motion with
  density-dependent, spatially structured noise, generalising
  the random-parameter SLWR of
  \citet{fan2023dynamically, fan2024stochastic} from static parametric
  variability to dynamic stochastic forcing.  This is the physical
  foundation from which the entire generative structure emerges.

\item \textbf{Exact one-point Fokker--Planck equation and Probability
  Flow ODE} (\Cref{sec:FPE,sec:PFODE}).  We derive the exact forward
  equation for the one-point marginal density with an explicit
  conditional drift, make the required closure transparent, and
  construct the equivalent deterministic Probability Flow ODE that is
  compatible with automatic differentiation.  The probability flow is
  not an externally imposed denoising process; it is the deterministic
  equivalent of the stochastic conservation-law dynamics.  This is the theoretical result that enables the proposed deep learning
  integration.

\item \textbf{Physics-informed score-matching architecture}
  (\Cref{sec:method}).  We formulate a learning architecture
  pairing a score network with an auxiliary advection-closure module,
  trainable via denoising score matching and a Fokker--Planck residual
  loss.  Because the physics residual is derived from the stochastic
  traffic model rather than from a generic diffusion prior, the
  learned distributional estimate is constrained by
  conservation-law dynamics. Previous transportation deep-learning formulations have rarely used a stochastic conservation-law distributional equation as the governing physics constraint (they used a deterministic PDE residual).
\end{enumerate}

\subsection{Paper organisation}

The remainder of the paper is organised as follows.
\Cref{sec:lit} reviews the literature on stochastic traffic modelling,
traffic state estimation, physics-informed neural networks, and
score-based generative modelling, and positions the present work
relative to the closest prior studies.
\Cref{sec:model} formulates the It\^{o}-type SLWR model and states
the standing assumptions.
\Cref{sec:FPE} derives the exact one-point Fokker--Planck equation
with complete proofs.
\Cref{sec:PFODE} derives the equivalent probability flow ODE,
establishes well-posedness, and provides physical interpretation of its
three constituent terms.
\Cref{sec:method} presents the physics-informed score matching
architecture, including the score network, the advection-closure
module, the loss function, training algorithm, and inference procedure.
\Cref{sec:discussion} discusses practical implications, limitations,
extensions to the stochastic fundamental diagram at both link and
network levels, and outlines the numerical validation programme
currently underway.
\Cref{sec:conclusion} concludes.

\section{Literature Review}\label{sec:lit}

We organise the review along four threads: stochastic traffic flow
modelling, traffic state estimation, physics-informed machine learning,
and score-based generative modelling, all of which converge in the present
work.  Each subsection concludes with an explicit identification of the
gap addressed by this paper.

\subsection{Stochastic macroscopic traffic flow modelling}
\label{sec:lit-stochastic}

The deterministic LWR model \citep{lighthill1955kinematic,
richards1956shock} is the foundational first-order macroscopic traffic
flow model.  Formulated as a scalar hyperbolic conservation law, it
captures kinematic shock and rarefaction waves but treats all traffic
variables as ensemble averages \citep{prigogine1971kinetic}.
Second-order extensions, including the Payne--Whitham model
\citep{payne1971model, whitham1974linear} and the Aw--Rascle--Zhang
(ARZ) model \citep{aw2000resurrection, zhang2002non}, introduce a
momentum equation that can reproduce non-equilibrium phenomena such as
stop-and-go waves.  The generic second-order modelling (GSOM)
framework of \citet{lebacque2007generic} unifies many such models by
coupling the LWR conservation law with dynamics of driver-specific
attributes.  Despite their descriptive power, these models are
deterministic and cannot capture stochastic variability inherent in
real traffic.

\paragraph{Kinetic theory and the Prigogine--Herman programme.}
A distinct tradition of modelling traffic stochasticity originates
from the kinetic theory of \citet{prigogine1971kinetic}, later
refined by \citet{paveri1975boltzmann} and surveyed by
\citet{klar1996enskog}.  In this framework, the fundamental object is
a \emph{velocity distribution function} $f(x,v,t)$, defined as the probability
density of finding a vehicle at position $x$ with speed $v$ at time
$t$.  Its evolution is governed by a Boltzmann-type
integro-differential equation that accounts for vehicle interactions
(acceleration, deceleration, and passing).  The macroscopic LWR model
can be recovered as the zeroth-moment closure of the kinetic equation,
by integrating out the velocity variable.

The Prigogine--Herman framework shares a conceptual affinity with the
present work: both govern a probability density that evolves according
to a physically derived equation.  However, the two frameworks operate
in fundamentally different spaces and serve different purposes.  In
kinetic theory, the probability is distributed over the \emph{speed}
variable $v$ and describes the heterogeneity of vehicle speeds within
a traffic stream at a given location; the stochasticity arises from
vehicle-to-vehicle interactions at the mesoscopic scale.  In our
framework, the probability is distributed over the \emph{macroscopic
density} variable $\hat\rho$ and quantifies the uncertainty in the
aggregate traffic state itself; the stochasticity is generated by
exogenous Brownian forcing on the conservation law at the macroscopic
scale.  The governing equation in kinetic theory is a Boltzmann-type
equation in $(x,v,t)$-space, whereas ours is a Fokker--Planck equation
in $(\hat\rho,x,t)$-space.  Perhaps most importantly, the purpose of
kinetic theory is to \emph{derive} macroscopic traffic models from
microscopic principles, whereas our purpose is to \emph{estimate}
macroscopic traffic states from sparse sensor data, using the
probability distribution as the vehicle for uncertainty quantification.

Incorporating stochasticity into macroscopic traffic models has been
pursued along several directions.

\paragraph{Randomised fundamental diagrams and noise-augmented models.}
\citet{li2012analysis} modified the speed--density function by
including a random noise term and developed an extended LWR model
based on these random fundamental relationships.
\citet{ngoduy2011multiclass} introduced stochastic fundamental
diagrams in a multiclass first-order setting.
\citet{sumalee2011stochastic} proposed a stochastic cell transmission
model (SCTM) using a zero-mean Gaussian random process to form a
probabilistic density.  \citet{boel2006compositional} developed a
compositional stochastic model for real-time freeway simulation.
A common criticism of these approaches is that adding noise terms
directly to a deterministic equation can produce negative densities
and lead to mean dynamics inconsistent with the deterministic
counterpart \citep{jabari2012stochastic}.

\paragraph{Stochastic models with explicit physical attribution.}
To address the consistency issue, \citet{jabari2012stochastic}
developed a cell-transmission-based stochastic model in which the
source of randomness is the uncertainty inherent in driver gap choice,
represented by random state-dependent vehicle time headways.
\citet{jabari2014probabilistic} subsequently derived probabilistic
macroscopic traffic flow relations from Newell's simplified
car-following model, treating time headways and spacings as random
variables.  \citet{martinez2020stochastic} introduced heterogeneous
(vehicle-dependent) jam densities to formulate a stochastic LWR model
in Lagrangian coordinates.

\paragraph{Random-parameter SLWR models.}
\citet{fan2023dynamically} proposed a stochastic LWR (SLWR) model that
randomises the free-flow speed $u_f(\omega)$ to represent driver
heterogeneity, while holding each individual driver's behaviour
constant along the road.  The model is formulated as a conservation
law parameterised by $u_f(\omega)$; once $\omega$ is sampled, the
equation reduces to a fully deterministic PDE.  The dynamically
bi-orthogonal (DyBO) method, based on Karhunen--Lo\`{e}ve expansion
and Hermite polynomial representation of the stochastic basis, was
applied to solve the model with approximately 5000$\times$ speedup
over the Monte Carlo method.  \citet{fan2024stochastic} extended this
framework to nonlinear speed--density relationships (e.g., Drake's
exponential model) by embedding a Taylor series expansion technique
within the DyBO formulation.

The Fan--Du framework provides a physically well-motivated and
computationally efficient approach to stochastic traffic modelling.
However, the scope of stochasticity it captures is inherently limited
by the static-parameter formulation.  First, because $u_f(\omega)$ is
time-invariant, dynamic perturbations (weather changes, incidents,
sensor noise) occurring \emph{during} the traffic process lie outside
the model's scope.  Second, the absence of a Brownian forcing term
means that no It\^{o} correction arises, there is no diffusion in
density space, and no Fokker--Planck equation governs the probability
distribution of density.  Ensemble statistics (mean and standard
deviation) are computed by integrating over the static parameter
distribution rather than by solving a distributional evolution equation.
The present work extends the Fan--Du framework by introducing dynamic
Brownian forcing, thereby enabling the Fokker--Planck and probability
flow ODE machinery developed in \Cref{sec:FPE,sec:PFODE}.  This
extension is consequential because the Fokker--Planck equation provides
an exact one-point density-space evolution law at each location,
with the unresolved spatial coupling isolated in a conditional drift
coefficient.  The associated probability flow ODE then furnishes a
deterministic transport representation that can be embedded as a
physics constraint in a transportation state-estimation model.

\paragraph{It\^{o}-type stochastic conservation laws.}
In the broader PDE literature, stochastic conservation laws driven by
Brownian motion have been studied extensively for well-posedness.
\citet{kim2003stochastic} established existence for a stochastic
scalar conservation law.
\citet{debussche2010scalar} proved existence and uniqueness of kinetic
solutions.  \citet{hofmanova2013degenerate} treated degenerate
parabolic SPDEs using kinetic formulations.
\citet{chen2012well} studied well-posedness of stochastic conservation
laws with spatially inhomogeneous flux.  These works provide the
mathematical foundation for It\^{o}-type stochastic LWR models but
do not address estimation from data or machine learning integration.

\paragraph{Gap addressed.}
We are not aware of prior work that formulates an It\^{o}-type
stochastic LWR model with density-dependent, spatially structured
Brownian noise and links it to a data-driven estimation framework.
Our model extends the Fan--Du random-parameter SLWR (the
latter being obtained when the noise coefficients $\sigma_k$ are set
to zero and a separate random parameterisation of the flux is
introduced) while enabling the derivation of a Fokker--Planck equation
and probability flow ODE.

\subsection{Traffic state estimation}\label{sec:lit-TSE}

Traffic state estimation (TSE) reconstructs spatio-temporal traffic
variables from sparse sensor data.  \citet{seo2017traffic} provided
a comprehensive survey categorising TSE methods as model-driven,
data-driven, and hybrid.

\paragraph{Model-driven methods.}
Classical approaches apply Kalman filtering and its variants on top
of macroscopic traffic models.  \citet{wang2005real} developed an
extended Kalman filter (EKF) for the METANET model.
\citet{wang2008real} demonstrated real-data testing of an adaptive EKF
with the second-order model, and \citet{wang2009adaptive} proposed an
adaptive freeway traffic state estimator.  \citet{di2010hybrid}
combined EKF with GPS data for arterial density estimation.
\citet{seo2017arz} developed an ARZ-based data fusion method.
These methods are principled but depend heavily on the accuracy of the
underlying deterministic model and on Gaussian assumptions that may
not hold in congested regimes.  Particle filters
\citep{mihaylova2007particle} relax the Gaussian assumption but incur
substantially higher computational cost.

\paragraph{Data-driven methods.}
Neural networks \citep{zheng2019deep}, convolutional and recurrent
architectures \citep{ma2017learning, wu2018hybrid}, and Gaussian
processes \citep{rodrigues2019multi} have been applied to TSE and
short-term traffic prediction.  \citet{adeli2004mesoscopic} used
wavelet-based neural networks for freeway work zone modelling, and
\citet{ghosh2006neural} developed neural network-wavelet
microsimulation models.  Deep learning methods can capture complex
spatio-temporal patterns but require large training datasets, offer no
guarantees of physical consistency, and may produce density fields
that violate conservation laws.

\paragraph{Hybrid methods: physics-informed deep learning (PIDL).}
\citet{shi2021physics} introduced a PIDL framework that encodes the
Greenshields-based ARZ model into a physics-informed neural network
(PINN) to regularise a data-driven estimation network.  The
architecture consists of a physics-uninformed neural network (PUNN)
mapping $(x,t)$ to $(\hat\rho,\hat u)$, and a PINN computing PDE
residuals $(\hat f_1, \hat f_2)$ as regularisation.  The loss function
combines data discrepancy (MSE on observations) and physical
discrepancy (MSE on PDE residuals at auxiliary collocation points).
Tested on the NGSIM US-101 dataset with both loop detector and probe
vehicle configurations, PIDL outperformed pure neural networks and
EKF in both estimation accuracy and data efficiency.  The framework
also demonstrated the ability to discover unknown model parameters
($\rho_{\max}$, $u_{\max}$, $\tau$) simultaneously with state
estimation.  Subsequent work applied PIDL with data-fitted flux
functions \citep{fan2013data} and extended the approach to
network-level estimation.

\paragraph{Gap addressed.}
All existing TSE methods, including PIDL, produce point
estimates of traffic states.  Uncertainty quantification, when present,
is applied post-hoc through techniques such as Monte Carlo dropout or
ensemble methods, which are not grounded in the underlying traffic
physics.  Our framework produces a data-conditioned distributional
estimate trained to approximately satisfy the one-point Fokker--Planck
equation as a physics regulariser, providing physically grounded
uncertainty quantification as a structural feature rather than a
post-hoc addition.

\subsection{Physics-informed neural networks}\label{sec:lit-PINN}

Physics-informed neural networks (PINNs) embed differential equations
into the loss function of neural networks, enabling the solution of
both forward and inverse problems for PDEs with limited data.  The
foundational work of \citet{raissi2019physics} demonstrated the
approach across fluid mechanics, quantum mechanics, and reaction-diffusion systems.  PINNs have since been applied in heat transfer
\citep{cai2021physics}, materials science, and geophysics.

In transportation, beyond the PIDL work of \citet{shi2021physics},
PINNs have been used for car-following model calibration
\citep{mo2021physics}, traffic flow forecasting with conservation law
constraints \citep{huang2022physics}, fundamental diagram estimation,
and network-level traffic modelling.
\citet{nicodemus2022physics} applied PINNs to the Aw--Rascle--Zhang
model for forward simulation.
\citet{lu2023physics} developed a differentiable programming framework
for joint traffic state and queue profile estimation on layered
computational graphs, demonstrating that representing deterministic
traffic flow PDEs on a computational graph enables efficient
gradient computation via automatic differentiation for traffic
estimation.  All of these approaches, however, embed
\emph{deterministic} PDEs and produce deterministic, point-valued outputs.  Several
extensions have been proposed for uncertainty quantification.
\citet{yang2021b} developed Bayesian PINNs (B-PINNs) that place
priors on network weights and perform approximate Bayesian inference,
capturing epistemic uncertainty.  Ensemble PINNs aggregate predictions
from multiple independently trained networks.
\citet{zhang2019quantifying} used latent-variable models within the
PINN framework.  However, these approaches treat uncertainty as a
property of the \emph{network} (epistemic), not as an intrinsic
feature of the \emph{physical system} (aleatoric).  The distinction is
fundamental: epistemic uncertainty can in principle be reduced with
more data, whereas aleatoric uncertainty arises from the stochastic
nature of the governing dynamics and persists regardless of data
volume.

\paragraph{Gap addressed.}
Our framework embeds a \emph{stochastic} PDE (through its equivalent
Fokker--Planck equation) into the PINN loss.  The uncertainty in the
output distribution is therefore aleatoric: it arises from the
stochastic dynamics of traffic flow itself, and is governed by the
physics, not by network parameter uncertainty.  This represents a
qualitatively different approach to UQ in physics-informed learning.

\subsection{Score-based generative modelling and probability flow ODEs}
\label{sec:lit-score}

Score-based generative models learn the score function
$\nabla_{\mathbf{x}}\log p_t(\mathbf{x})$ of a time-dependent data
distribution and use it to generate samples via stochastic differential
equations (SDEs) or the equivalent deterministic probability flow ODE.
\citet{song2019generative} introduced score matching with Langevin
dynamics for generative modelling, and \citet{ho2020denoising}
developed denoising diffusion probabilistic models (DDPMs).  The
unifying framework of \citet{song2021score} established that for any
It\^{o} SDE with drift $\mathbf{f}(\mathbf{x},t)$ and diffusion
$g(t)$, there exists a deterministic ODE, the probability flow
ODE, whose marginal distributions match those of the SDE at every
time.  In the notation of \citet{song2021score}, this ODE reads
\begin{equation}\label{eq:song-PF}
  \dd\mathbf{x} = \bigl[\mathbf{f}(\mathbf{x},t)
  - \tfrac{1}{2}g(t)^2\nabla_{\mathbf{x}}\log p_t(\mathbf{x})\bigr]\dd t.
\end{equation}
This SDE/ODE equivalence has become the theoretical backbone of modern
diffusion models and has found applications in image generation
\citep{dhariwal2021diffusion}, molecular design, audio synthesis, and
scientific computing.

In the mathematical community, probability flow ODEs have been
connected to optimal transport theory.  The Benamou--Brenier formula
\citep{benamou2000computational} characterises the optimal transport
map as the velocity field minimising kinetic energy subject to the
continuity equation constraint.  Flow matching
\citep{lipman2023flow, albergo2023building} directly parameterises
conditional vector fields along interpolation paths between noise and
data.  Rectified flow \citep{liu2023flow} iteratively straightens
transport paths to enable few-step generation.  Schr\"{o}dinger
bridges \citep{de2021diffusion, chen2022likelihood} generalise optimal
transport to the stochastic setting by finding the most likely
diffusion process connecting two marginal distributions.

\paragraph{Neural ODEs and continuous normalising flows.}
A closely related development is the neural ODE framework of
\citet{chen2018neural}, which parameterises the dynamics of a hidden
state by a neural network and solves the resulting ODE with adaptive
solvers, enabling backpropagation through the dynamics via the adjoint
method.  When the neural ODE is interpreted as a transport map acting
on probability densities, one obtains a \emph{continuous normalising
flow} (CNF): a learned velocity field whose flow pushes a simple base
distribution to a complex target, with the log-density evolving via
the instantaneous change-of-variables formula
$\partial_t \log p = -\nabla\cdot v_\theta$
\citep{chen2018neural,grathwohl2019ffjord}.  CNFs are structurally
analogous to the probability flow ODE~\eqref{eq:song-PF}: both define
a deterministic ODE whose flow transports probability mass.  The
difference is that in CNFs the velocity field is a free neural network,
whereas in score-based models it is constrained by the score of the
evolving density.  Our framework shares the CNF perspective---the
closed PF-ODE~\eqref{eq:PFODE-closed} is a continuity equation with a
velocity field comprising a physics-derived drift closure and a learned
score---but differs in two respects: (i)~the velocity field is not
freely parameterised but is structurally constrained by the stochastic
conservation-law physics, and (ii)~the evolution variable is not a
latent code but the physical traffic density $\hat\rho$, with spatial
location $x$ entering as an external parameter.

Despite this mathematical maturity, probability flow ODEs have not
been applied to traffic state estimation.  The principal obstacle is
that the standard theory assumes a purely parabolic diffusion process,
whereas traffic dynamics are governed by \emph{hyperbolic} conservation
laws.  The drift in the traffic setting contains $\partial_x f(\rho)$,
a spatial derivative of the flux function, which introduces
characteristic wave propagation absent in the standard setting.

\paragraph{Gap addressed.}
We derive a probability flow ODE for a stochastic
conservation law, in which parabolic evolution in density space is
parameterised by physical location with the hyperbolic traffic physics
retained implicitly through a conditional drift closure.  Prior transportation studies have not combined an exact one-point forward equation, explicit conditional-drift closure, and probability-flow representation into a score-based physics residual for stochastic traffic estimation.

\subsection{The macroscopic fundamental diagram and its stochastic
variability}\label{sec:lit-MFD}

The idea that a well-defined relationship exists between aggregate
flow and aggregate density at the macroscopic level has been present
in traffic flow theory for decades, with roots in the kinetic theory
of \citet{prigogine1971kinetic} and early treatments of two-fluid
models.  The concept was reinvigorated at the \emph{network} level by
\citet{geroliminis2008existence}, who provided large-scale empirical
evidence for a reproducible network-level MFD using detector data from
Yokohama, Japan.  Since then, the network MFD has been adopted
as a tool for perimeter control, congestion pricing, and regional
traffic monitoring.

A persistent feature of the empirical MFD is \emph{scatter}: at the
same average density, different realisations of traffic produce
different aggregate flows.  \citet{mazloumian2010spatial} showed that
this scatter is largely explained by the spatial variability of vehicle
densities across the network.  \citet{leclercq2014macroscopic} compared
estimation methods and confirmed the sensitivity of the MFD shape to
the spatial distribution of congestion.  \citet{ambuehl2023understanding}
connected congestion propagation patterns to MFD hysteresis using
percolation theory.

Despite the empirical recognition that MFD scatter is not random noise
but a systematic consequence of spatial heterogeneity, no principled
stochastic framework currently exists for the MFD.  Existing treatments
either fit statistical models (e.g., variance as a function of
$\bar\rho$) without physical grounding, or simulate scatter via
microsimulation without connecting it to macroscopic stochastic dynamics.
The framework developed in this paper provides the mathematical
machinery to fill this gap: the spatially structured Brownian forcing
in the It\^{o}-type SLWR generates precisely the kind of spatial
density heterogeneity that drives MFD scatter, and the Fokker--Planck
equation can in principle be aggregated to yield a distributional
evolution law at the network level.  We elaborate on this extension
direction in \Cref{sec:discussion}.

\subsection{Summary of research gaps and contributions}
\label{sec:lit-summary}

The foregoing review reveals three interrelated gaps in the literature
that this paper addresses.

\paragraph{Gap 1: The absence of a distributional evolution equation
for macroscopic traffic density.}
Stochastic traffic flow models fall into two categories.  On the one
hand, random-parameter models such as the Fan--Du SLWR
\citep{fan2023dynamically, fan2024stochastic} treat uncertainty as
static parametric variability: once the random parameter is drawn, the
system is deterministic, and no equation governs how the probability
distribution of density evolves in time.  On the other hand, the
Prigogine--Herman kinetic theory \citep{prigogine1971kinetic,
paveri1975boltzmann} does provide a distributional evolution equation
(the Boltzmann equation for the speed distribution), but it operates
at the mesoscopic scale and in speed space, not at the macroscopic
scale in density space.  Meanwhile, It\^{o}-type stochastic
conservation laws have been studied in the PDE theory community
\citep{kim2003stochastic, debussche2010scalar, hofmanova2013degenerate},
but purely for well-posedness; they have never been connected to
traffic estimation or machine learning.  This paper addresses that
gap by formulating an It\^{o}-type stochastic LWR model and deriving
the Fokker--Planck equation that governs the one-point marginal
density of traffic density (\Cref{sec:FPE}).  This equation
provides an exact one-point forward law for the macroscopic traffic-density
distribution under dynamic stochastic forcing, together with a clear
statement of the closure required for computation.

\paragraph{Gap 2: The lack of a deterministic equivalent that can
serve as a physics constraint in neural networks.}
Even if a stochastic traffic model were available, embedding it
directly into a physics-informed neural network would be impractical:
stochastic PDEs require sampling-based solvers (e.g., Monte Carlo),
which are incompatible with the backpropagation-based training of
neural networks.  The probability flow ODE, introduced by
\citet{song2021score} in the context of generative modelling, provides
a deterministic ODE whose marginal distributions match those of the
original SDE, making it amenable to PINN-style residual evaluation
via automatic differentiation.  However, the existing theory assumes
purely parabolic diffusion processes and does not apply to hyperbolic
conservation laws.  Here the probability-flow construction is adapted to the exact one-point forward equation derived from a stochastic conservation law, yielding a probability flow ODE
(\Cref{sec:PFODE}) in which the density-space evolution is parabolic
while the hyperbolic traffic physics is retained implicitly through a
conditional drift closure.  Once the closure is specified, the
resulting ODE can be evaluated pointwise and differentiated through a
neural network, enabling distributional physics constraints that were
previously inaccessible.

\paragraph{Gap 3: Traffic state estimation methods produce point
estimates without physically grounded uncertainty quantification.}
All existing TSE methods (model-driven \citep{wang2008real},
data-driven \citep{zheng2019deep}, and hybrid
\citep{shi2021physics}) output a single best estimate of traffic
density and velocity.  Uncertainty quantification, when attempted, is
applied post-hoc through techniques such as Monte Carlo dropout,
ensemble methods, or Bayesian neural network approximations
\citep{yang2021b}.  These approaches capture epistemic uncertainty
(arising from limited data or model capacity) but not aleatoric
uncertainty (arising from the intrinsic stochasticity of traffic
dynamics).  Moreover, the resulting uncertainty estimates are not
required to satisfy any traffic physics.  We address
this gap by designing a physics-informed score matching framework
(\Cref{sec:method}) in which the probability flow ODE serves as a
distributional physics constraint during training.  The output is a
distributional estimate
$p_\theta(\hat\rho; x, t)$, trained to approximately satisfy
the Fokker--Planck equation.  Point estimates, credible
intervals, and risk metrics (e.g., the probability that density
exceeds a congestion threshold) are derived as functionals of this
estimate.

\paragraph{Relation among the three contributions.}
The three contributions are linked.  The It\^{o}-type SLWR
model (Contribution~1) yields a Fokker--Planck equation that is absent
from the static-parameter formulation of Fan et al.  The probability flow ODE (Contribution~2) converts that
equation into a deterministic form compatible with neural network
training, while keeping the drift closure explicit.  The
physics-informed score-matching method (Contribution~3) then uses
the resulting probability flow ODE as a distributional regulariser for data-conditioned distributional estimation from sparse
sensor data.

\Cref{tab:positioning} provides a dimension-by-dimension comparison with the two closest prior works.

\begin{table}[ht]
\centering
\caption{Positioning of the proposed framework relative to the closest
prior works.}
\label{tab:positioning}
\renewcommand{\arraystretch}{1.25}
\small
\begin{tabular}{p{2.8cm}|p{3.0cm}|p{3.0cm}|p{3.5cm}}
\toprule
& \textbf{Fan--Du SLWR}
  \newline{\scriptsize CACIE 2023; Trans.\ B 2024}
& \textbf{Shi--Di PIDL}
  \newline{\scriptsize AAAI 2021}
& \textbf{This paper} \\
\midrule
Stochastic model & Random-parameter PDE & None (det.\ ARZ) & It\^{o} SPDE \\
Distributional eq. & None & N/A & Fokker--Planck \\
Physics in learning & Forward solver & Det.\ PDE residual & PF-ODE residual \\
Output & Mean \& std & Point estimate & Full $p_\theta(\hat\rho;x,t)$ \\
UQ mechanism & Parametric & Post-hoc only & Intrinsic (score) \\
Dynamic noise & No & No & Yes ($\dd W_t$) \\
Sensor data fusion & Not designed & Core capability & Core capability \\
\bottomrule
\end{tabular}
\end{table}

\section{Model Formulation}\label{sec:model}

\subsection{Probability space and spatial domain}

Let $(\Omega,\mathcal{F},\{\mathcal{F}_t\}_{t\ge 0},\PP)$ be a
complete filtered probability space satisfying the usual conditions
(right-continuity and completeness).  Fix a bounded highway segment
$\mathcal{D}=[0,L]\subset\RR$ and a time horizon $T>0$.  Write
$Q_T:=\mathcal{D}\times(0,T]$.  Let $\{W^k_t\}_{k=1}^K$ be $K$
mutually independent standard $(\mathcal{F}_t)$-Brownian motions.

\subsection{The It\^{o}-type stochastic LWR model}

\begin{definition}[It\^{o}-type SLWR with finite-dimensional noise]\label{def:SLWR}
The \emph{stochastic Lighthill--Whitham--Richards} (SLWR) equation reads
\begin{equation}\label{eq:SLWR}
  \dd\rho(x,t) = -\partial_x f\bigl(\rho(x,t)\bigr)\,\dd t
  + \sum_{k=1}^K \sigma_k\bigl(\rho(x,t)\bigr)\,e_k(x)\,\dd W^k_t,
  \qquad (x,t)\in Q_T,
\end{equation}
with initial data $\rho(x,0)=\rho_0(x)$ and boundary conditions on
$\partial\mathcal{D}$ appropriate for the traffic scenario.
\end{definition}

\begin{remark}[On the form of the stochastic forcing]
\label{rem:non-conservative-noise}
The noise in~\eqref{eq:SLWR} enters as a source-type forcing rather
than in divergence form.  Consequently, pathwise conservation of the
total vehicle count $\int_\mathcal{D}\rho\,\dd x$ is not preserved:
integrating~\eqref{eq:SLWR} over $\mathcal{D}$ yields a stochastic
fluctuation $\sum_k(\int\sigma_k(\rho)\,e_k(x)\,\dd x)\,\dd W^k_t$
in addition to the boundary flux.  Physically, this formulation
represents exogenous random perturbations---demand fluctuations,
incidents, weather-induced capacity changes, and on/off-ramp
variability---that act as local sources and sinks of density on a
bounded segment.  Such perturbations are ubiquitous in real freeway
corridors.  At the level of expected mass, integrating over
$\mathcal{D}$ and taking expectations yields
$\partial_t\int_\mathcal{D}\EE[\rho]\,\dd x
= -\EE[f(\rho)]|_{\partial\mathcal{D}}$
under appropriate integrability conditions, so the mean mass balance
retains the conservation-law structure up to boundary flux.
The endpoint vanishing condition
$\sigma_k(0)=\sigma_k(\rho_{\max})=0$
(Assumption~\ref{ass:noise}(i)) ensures that the noise intensity
vanishes at the physical density bounds; state-space invariance of
the interval $[0,\rho_{\max}]$ is imposed directly via
Assumption~\ref{ass:smooth}(iii) rather than derived from the endpoint
condition alone.  An alternative formulation placing the noise inside
the spatial divergence, $\dd\rho = -\partial_x[f(\rho)\,\dd t +
\sum_k\sigma_k(\rho)\,e_k(x)\,\dd W^k_t]$, would preserve pathwise
conservation but introduces a fundamentally different SPDE whose
analysis requires separate treatment.  We adopt the source-type
formulation here because it yields a tractable one-point forward
equation with standard It\^{o} calculus, while the resulting
non-conservation is physically interpretable and bounded by the noise
amplitude.
\end{remark}

The model components are governed by five standing assumptions.

\begin{assumption}[Flux function]\label{ass:flux}
$f\in C^2([0,\rho_{\max}])$ with $f(\rho)=\rho\,v(\rho)$, where
$v:[0,\rho_{\max}]\to[0,u_f]$ is a decreasing speed--density relation
satisfying $v(0)=u_f>0$ and $v(\rho_{\max})=0$.  The flux is concave:
$f''(\rho)\le 0$ on $[0,\rho_{\max}]$ with $f''(\rho)<0$ for a.e.\
$\rho$.
\end{assumption}

\begin{remark}[Generality of the flux]\label{rem:flux-general}
Assumption~\ref{ass:flux} is satisfied by smooth concave fundamental
diagrams, including Greenshields
($f(\rho)=u_f\rho(1-\rho/\rho_{\max})$) and Drake
($f(\rho)=u_f\rho\exp(-\rho^2/(2k_o^2))$, cf.\ Fan et al., 2024).
Piecewise-linear models such as the triangular fundamental diagram do
not satisfy the $C^2$ requirement; extending the theory to
piecewise-smooth flux functions would require a separate treatment.
\end{remark}

\begin{assumption}[Noise structure]\label{ass:noise}
For each $k=1,\dots,K$:
\begin{enumerate}[label=\textup{(\roman*)}]
  \item $\sigma_k\in C^3([0,\rho_{\max}])$ with
    $\sigma_k(0)=\sigma_k(\rho_{\max})=0$;
  \item $e_k\in C^3(\overline{\mathcal{D}})$ are deterministic spatial
    basis functions;
  \item (Non-degeneracy)
    $\Sigma^2(\rho,x):=\sum_{k=1}^K\sigma_k(\rho)^2 e_k(x)^2>0$
    for all $(\rho,x)\in(0,\rho_{\max})\times\mathcal{D}$.
\end{enumerate}
\end{assumption}

\begin{remark}[Regularity of the noise coefficients]
\label{rem:noise-regularity}
The $C^3$ regularity in Assumption~\ref{ass:noise}(i)--(ii) is
stronger than what is needed for the Fokker--Planck derivation alone
(which requires only $C^2$).  The additional derivative is needed in
the score-form residual~\eqref{eq:score-FPE-residual}: the term
$\partial_{\hat\rho}^2 v_{\theta,\phi}$ involves
$\partial_{\hat\rho}^3\Sigma^2$ through the expansion of the
probability-flow velocity.  In practice, typical noise
parameterisations such as
$\sigma_k(\rho)=\alpha_k\rho(\rho_{\max}-\rho)$ are $C^\infty$, so
this requirement is not restrictive.
\end{remark}

\begin{assumption}[Initial data]\label{ass:init}
$\rho_0\in C^1(\overline{\mathcal{D}})$ with
$0<\rho_0(x)<\rho_{\max}$ for all $x\in\mathcal{D}$.
\end{assumption}

\begin{assumption}[Regularity regime]\label{ass:smooth}
The solution $\rho(\cdot,\cdot)$ of~\eqref{eq:SLWR} satisfies:
(i)~for each $t$, the map $x\mapsto\rho(x,t)$ is in $C^1(\overline{\mathcal{D}})$
$\PP$-a.s.; (ii)~for each $x$, the process $t\mapsto\rho(x,t)$ is a
continuous semimartingale; (iii)~$0<\rho(x,t)<\rho_{\max}$ for all
$(x,t)\in Q_T$, $\PP$-a.s.\ (no shock formation during $[0,T]$).
\end{assumption}

\begin{remark}[Practical scope of Assumption~\ref{ass:smooth}]
\label{rem:smooth-scope}
This assumption is satisfied on the entire estimation horizon for
sub-critical (free-flow) traffic states with smooth initial data, and
for arbitrary smooth initial data over sufficiently short time intervals
before the first shock forms.  For traffic scenarios that include shocks,
the viscous regularisation discussed in \Cref{rem:shock} below
provides a systematic extension in which Assumption~\ref{ass:smooth}
holds globally (Section 5).
\end{remark}

\begin{assumption}[Existence of a smooth marginal density]\label{ass:density}
For each $(x,t)\in Q_T$, the random variable $\rho(x,t)$ possesses a
probability density $p(\hat\rho;x,t)$ with respect to Lebesgue measure
on $(0,\rho_{\max})$, and
$p\in C^{2,1,1}((0,\rho_{\max})\times Q_T)$.
\end{assumption}

\begin{remark}[On Assumption~\ref{ass:density}]
The existence of a smooth marginal density for finite-dimensional SDEs
with non-degenerate diffusion is motivated by Malliavin-calculus
techniques.  In the present setting, where the pointwise process
$Y_t^x$ inherits its noise from a spatially projected
finite-dimensional Brownian motion, the non-degeneracy condition in
Assumption~\ref{ass:noise}(iii) suggests that a smooth one-point
density should exist, but verifying this for the full SPDE
point-evaluation requires care that goes beyond standard
finite-dimensional results.  We therefore impose
Assumption~\ref{ass:density} directly rather than deriving it.
\end{remark}

\subsection{Aggregate diffusion coefficient}

Define the aggregate diffusion coefficient
\begin{equation}\label{eq:D-def}
  D(\rho,x) := \frac{1}{2}\sum_{k=1}^K
  \bigl[\sigma_k(\rho)\,e_k(x)\bigr]^2
  = \frac{1}{2}\,\Sigma^2(\rho,x).
\end{equation}

\subsection{Distinction from the Fan--Du random-parameter SLWR}

The SLWR of \citet{fan2023dynamically,fan2024stochastic} reads
$\partial_t k + \partial_x[k\,v(k;u_f(\omega))]=0$ where
$u_f(\omega)$ is a \emph{time-invariant, spatially homogeneous}
random parameter.  Once $\omega$ is sampled, the equation is
deterministic.  Key structural differences:
\begin{enumerate}[label=(\alph*)]
  \item Equation~\eqref{eq:SLWR} contains Brownian motion
    $\dd W^k_t$; Fan--Du does not.
  \item No Fokker--Planck equation arises in the Fan--Du framework
    (no diffusion mechanism in density space).
  \item Setting $\sigma_k\equiv 0$ in~\eqref{eq:SLWR} reduces it to a
    deterministic conservation law; with a separate random-parameter
    specification of the flux $f(\rho;u_f(\omega))$, this connects to
    the Fan--Du formulation.
    Our model \emph{extends} that framework by introducing dynamic
    Brownian forcing.
\end{enumerate}

\subsection{Summary of notation}

For the reader's convenience, \Cref{tab:notation} collects the
principal symbols used in the mathematical development
(\Cref{sec:FPE,sec:PFODE,sec:method}).

\begin{table}[ht]
\centering
\caption{Summary of notation.}
\label{tab:notation}
\renewcommand{\arraystretch}{1.25}
\small
\begin{tabular}{cl}
\toprule
\textbf{Symbol} & \textbf{Description} \\
\midrule
\multicolumn{2}{l}{\textit{Domain and probability space}} \\
$\mathcal{D}=[0,L]$ & Highway segment \\
$Q_T=\mathcal{D}\times(0,T]$ & Spatio-temporal domain \\
$(\Omega,\mathcal{F},\{\mathcal{F}_t\},\PP)$
  & Filtered probability space \\
$W^k_t$\;$(k=1,\dots,K)$ & Independent standard Brownian motions \\[4pt]
\multicolumn{2}{l}{\textit{Traffic variables and flux}} \\
$\rho(x,t)$ & Traffic density (random field) \\
$\hat\rho$ & Density-space coordinate (deterministic query variable) \\
$\rho_{\max}$ & Jam density \\
$\rho_0(x)$ & Initial density profile \\
$f(\rho)=\rho\,v(\rho)$ & Flux function (fundamental diagram) \\
$v(\rho)$ & Speed--density relation \\
$u_f$ & Free-flow speed \\[4pt]
\multicolumn{2}{l}{\textit{Noise structure}} \\
$\sigma_k(\rho)$ & Density-dependent noise intensity (mode $k$) \\
$e_k(x)$ & Spatial basis function (mode $k$) \\
$\Sigma^2(\rho,x)$ & Aggregate diffusion:
  $\sum_k\sigma_k(\rho)^2 e_k(x)^2$ \\
$D(\rho,x)$ & Half-diffusion: $\tfrac{1}{2}\Sigma^2(\rho,x)$ \\[4pt]
\multicolumn{2}{l}{\textit{Pointwise process and drift}} \\
$Y_t^x$ & Pointwise process: $\rho(x,t)$ at fixed $x$ \\
$\beta_t^x$ & Pathwise drift:
  $-f'(Y_t^x)\,\partial_x\rho(x,t)$ \\
$b(\hat\rho,x,t)$ & Conditional drift:
  $\EE[\beta_t^x\mid Y_t^x=\hat\rho]$ \\
$b_\phi(\hat\rho,x,t)$ & Learned advection-closure module \\[4pt]
\multicolumn{2}{l}{\textit{Distributional quantities}} \\
$p(\hat\rho;x,t)$ & One-point marginal density of $\rho(x,t)$ \\
$s(\hat\rho;x,t)$ & Score function:
  $\partial_{\hat\rho}\log p$ \\
$J(\hat\rho;x,t)$ & Probability flux at density-space boundary \\
$v_{\mathrm{PF}}$ & Probability-flow velocity (exact) \\[4pt]
\multicolumn{2}{l}{\textit{Learning architecture}} \\
$s_\theta$ & Score network (parameters $\theta$) \\
$v_{\theta,\phi}$ & Closed probability-flow velocity \\
$\mathcal{R}_{\theta,\phi}$ & Score-form FPE residual \\
$p_\theta$ & Reconstructed distributional estimate \\
$\mathcal{L}_{\text{SM}},\;
 \mathcal{L}_{\text{PF}},\;
 \mathcal{L}_{\text{BC}}$
  & Score matching, physics, and boundary losses \\
$\lambda,\;\lambda_{\text{BC}}$
  & Loss weighting hyperparameters \\
$\sigma_{\text{DSM}}$ & DSM perturbation noise scale \\
$\rho_*$ & Reference density for quadrature integration \\
\bottomrule
\end{tabular}
\end{table}

\section{The Fokker--Planck Equation}\label{sec:FPE}

\subsection{Pointwise semimartingale and conditional drift}

At each fixed spatial location $x\in\mathcal{D}$, define the scalar
process
\[
  Y_t^x := \rho(x,t).
\]
Under Assumption~\ref{ass:smooth}, $Y_t^x$ is a continuous semimartingale and satisfies
\begin{equation}\label{eq:scalar-SDE}
  \dd Y_t^x = \beta_t^x\,\dd t
  + \sum_{k=1}^K \sigma_k(Y_t^x)\,e_k(x)\,\dd W_t^k,
\end{equation}
with pathwise drift
\begin{equation}\label{eq:mu-def}
  \beta_t^x := -f'\!\left(Y_t^x\right)\,\partial_x\rho(x,t).
\end{equation}
The dependence on $\partial_x\rho$ shows that the pointwise process is
not closed when viewed in isolation.  In transportation terms, the local
stochastic density evolution still carries unresolved information about
the upstream and downstream traffic state through the conservation-law
gradient.  The correct one-point forward equation therefore involves a
conditional drift coefficient rather than the raw pathwise drift.

\begin{definition}[Conditional drift coefficient]\label{def:conditional-drift}
For $(\hat\rho,x,t)\in(0,\rho_{\max})\times Q_T$, define
\begin{equation}\label{eq:b-def}
  b(\hat\rho,x,t)
  := \EE\!\left[\beta_t^x\mid Y_t^x=\hat\rho\right]
  = \EE\!\left[-f'\!\left(\rho(x,t)\right)\partial_x\rho(x,t)
    \mid \rho(x,t)=\hat\rho\right].
\end{equation}
\end{definition}

\begin{remark}[Transportation interpretation]\label{rem:closure}
The coefficient $b(\hat\rho,x,t)$ is the one-point imprint of the LWR
advection term.  It is exact, but it is not closed at the one-point
level unless it is approximated, modelled, or learned.  
\end{remark}

\subsection{Statement and proof}

\begin{proposition}[Exact one-point Fokker--Planck equation]\label{prop:FPE}
Suppose Assumptions~\ref{ass:flux}--\ref{ass:density} hold and that the
conditional drift coefficient $b$ in~\eqref{eq:b-def} is measurable and
locally bounded. Then, for each fixed $x\in\mathcal{D}$, the one-point
marginal density $p(\hat\rho;x,t)$ of $Y_t^x=\rho(x,t)$ satisfies
\begin{equation}\label{eq:FPE}
  \partial_t p(\hat\rho;x,t)
  = -\partial_{\hat\rho}\!\left[b(\hat\rho,x,t)\,p(\hat\rho;x,t)\right]
  + \frac{1}{2}\partial_{\hat\rho}^2\!\left[\Sigma^2(\hat\rho,x)
    \,p(\hat\rho;x,t)\right],
\end{equation}
with initial condition
\begin{equation}\label{eq:FPE-initial}
  p(\hat\rho;x,0)=\delta\!\left(\hat\rho-\rho_0(x)\right),
\end{equation}
understood in the distributional sense (equivalently, one may begin
from a mollified initial law $\delta_\epsilon(\hat\rho-\rho_0(x))$ and
let $\epsilon\downarrow 0$), and zero probability-flux boundary
condition
\begin{equation}\label{eq:zero-flux-bc}
  J(\hat\rho;x,t)
  := b(\hat\rho,x,t)\,p(\hat\rho;x,t)
   - \frac{1}{2}\partial_{\hat\rho}\!\left[\Sigma^2(\hat\rho,x)
     \,p(\hat\rho;x,t)\right] = 0,
  \qquad \hat\rho\in\{0,\rho_{\max}\}.
\end{equation}
\end{proposition}

\begin{proof}
Fix $x\in\mathcal{D}$ and let $Y_t:=Y_t^x$. For any test function
$\varphi\in C_c^2((0,\rho_{\max}))$, It\^{o}'s formula applied to
\eqref{eq:scalar-SDE} yields
\begin{equation}\label{eq:Ito-applied}
  \dd\varphi(Y_t)
  = \varphi'(Y_t)\,\beta_t\,\dd t
    + \frac{1}{2}\varphi''(Y_t)\,\Sigma^2(Y_t,x)\,\dd t
    + \varphi'(Y_t)\sum_{k=1}^K \sigma_k(Y_t)e_k(x)\,\dd W_t^k.
\end{equation}
Taking expectations and using the martingale property of the stochastic
integral,
\begin{equation}\label{eq:expect-identity}
  \frac{\dd}{\dd t}\EE[\varphi(Y_t)]
  = \EE\!\left[\varphi'(Y_t)\beta_t
    + \frac{1}{2}\varphi''(Y_t)\Sigma^2(Y_t,x)\right].
\end{equation}
By conditioning on $Y_t$ and using
Definition~\ref{def:conditional-drift},
\begin{align}\label{eq:weak-form}
  \EE[\varphi'(Y_t)\beta_t]
  &= \EE\!\left[\EE\!\left[\varphi'(Y_t)\beta_t\mid Y_t\right]\right]
   = \EE\!\left[\varphi'(Y_t)\,\EE[\beta_t\mid Y_t]\right] \notag\\
  &= \int_0^{\rho_{\max}} \varphi'(\hat\rho)
     \,b(\hat\rho,x,t)\,p(\hat\rho;x,t)\,\dd\hat\rho,
\end{align}
and similarly
\begin{equation}
  \EE\!\left[\frac{1}{2}\varphi''(Y_t)\Sigma^2(Y_t,x)\right]
  = \frac{1}{2}\int_0^{\rho_{\max}} \varphi''(\hat\rho)
    \,\Sigma^2(\hat\rho,x)\,p(\hat\rho;x,t)\,\dd\hat\rho.
\end{equation}
Therefore,
\begin{equation}
  \int_0^{\rho_{\max}} \varphi(\hat\rho)\,\partial_t p(\hat\rho;x,t)
  \,\dd\hat\rho
  = \int_0^{\rho_{\max}} \varphi'(\hat\rho)
    \,b(\hat\rho,x,t)\,p(\hat\rho;x,t)\,\dd\hat\rho
  + \frac{1}{2}\int_0^{\rho_{\max}} \varphi''(\hat\rho)
    \,\Sigma^2(\hat\rho,x)\,p(\hat\rho;x,t)\,\dd\hat\rho.
\end{equation}
Integrating by parts in $\hat\rho$ gives
\[
  \int_0^{\rho_{\max}} \varphi\,\partial_t p\,\dd\hat\rho
  = \int_0^{\rho_{\max}} \varphi
    \left[-\partial_{\hat\rho}(bp)
    + \frac{1}{2}\partial_{\hat\rho}^2(\Sigma^2 p)\right]\dd\hat\rho.
\]
Since this holds for all $\varphi\in C_c^2((0,\rho_{\max}))$,
\eqref{eq:FPE} follows in the distributional sense; the assumed
regularity upgrades it to the classical form. The zero-flux condition
\eqref{eq:zero-flux-bc} expresses conservation of total probability on
$[0,\rho_{\max}]$.
\end{proof}

\begin{remark}[On the regularity of the conditional drift]
\label{rem:b-regularity}
The hypothesis that $b(\hat\rho,x,t)$ is measurable and locally bounded
is a non-trivial regularity condition on the joint law of the SPDE
solution.  The pathwise drift
$\beta_t^x = -f'(\rho(x,t))\,\partial_x\rho(x,t)$ involves the
spatial gradient $\partial_x\rho$, which is itself a random field.
The conditional expectation
$b(\hat\rho,x,t) = \EE[-f'(\rho)\,\partial_x\rho \mid \rho(x,t)=\hat\rho]$
integrates over all spatial-gradient configurations consistent with
the local density equalling $\hat\rho$; its boundedness therefore
depends on the joint distribution of $(\rho,\partial_x\rho)$.
Under the smooth-regime Assumption~\ref{ass:smooth}, $\partial_x\rho$
is bounded pathwise and the conditional expectation inherits local
boundedness.  Verifying this condition from first principles in
broader regimes would require detailed regularity estimates on the
SPDE that lie outside the scope of this paper.
\end{remark}

\begin{remark}[Scope of the one-point equation]\label{rem:no-full-x-fpe}
Equation~\eqref{eq:FPE} is a one-point forward equation parametrised by
physical location $x$.  It is not, by itself, a closed PDE in
$(\hat\rho,x,t)$ with hyperbolic transport in $x$.  The conservation-law
structure enters through the conditional drift coefficient
$b(\hat\rho,x,t)$, which carries the unresolved spatial coupling.
\end{remark}

\begin{remark}[Practical closure for transportation methodology]
\label{rem:closure-options}
To obtain a computable surrogate for traffic state estimation, one may
introduce a closure of the form
\begin{equation}\label{eq:closure-general}
  b(\hat\rho,x,t) \approx b_\phi(\hat\rho,x,t),
\end{equation}
with either a direct auxiliary drift model $b_\phi$ or a structured
closure such as
\begin{equation}\label{eq:closure-m}
  b_\phi(\hat\rho,x,t) = -f'(\hat\rho)\,m_\phi(\hat\rho,x,t).
\end{equation}
A narrower mean-field approximation is
\begin{equation}\label{eq:closure-meanfield}
  b_\phi(\hat\rho,x,t) = -f'(\hat\rho)\,\partial_x\bar\rho_\phi(x,t),
\end{equation}
where $\bar\rho_\phi$ is an auxiliary mean-state model.  In the present
paper, this is best described as a transportation-oriented closure
assumption rather than part of the exact theorem.
\end{remark}

\section{The Probability Flow ODE}\label{sec:PFODE}

\subsection{Reformulation as a continuity equation}

\begin{lemma}[Continuity equation form of the one-point FPE]\label{lem:cont}
Assume $p(\hat\rho;x,t)>0$ on $(0,\rho_{\max})\times Q_T$ and define the
score
\begin{equation}\label{eq:score-def}
  s(\hat\rho;x,t):=\partial_{\hat\rho}\log p(\hat\rho;x,t).
\end{equation}
Then equation~\eqref{eq:FPE} is equivalent to
\begin{equation}\label{eq:continuity}
  \partial_t p + \partial_{\hat\rho}
  \bigl[v_{\mathrm{PF}}(\hat\rho,x,t)\,p\bigr] = 0,
\end{equation}
with probability-flow velocity
\begin{equation}\label{eq:vPF}
  v_{\mathrm{PF}}(\hat\rho,x,t)
  = b(\hat\rho,x,t)
    - \frac{1}{2}\partial_{\hat\rho}\Sigma^2(\hat\rho,x)
    - \frac{1}{2}\Sigma^2(\hat\rho,x)
      \partial_{\hat\rho}\log p(\hat\rho;x,t).
\end{equation}
\end{lemma}

\begin{proof}
Starting from~\eqref{eq:FPE}, use the product rule:
\begin{align}
  \frac{1}{2}\partial_{\hat\rho}^2[\Sigma^2 p]
  &= \frac{1}{2}\partial_{\hat\rho}\Bigl[
      (\partial_{\hat\rho}\Sigma^2)\,p
      + \Sigma^2\partial_{\hat\rho}p\Bigr] \notag\\
  &= \frac{1}{2}\partial_{\hat\rho}\Bigl[
      \bigl(\partial_{\hat\rho}\Sigma^2
      + \Sigma^2\partial_{\hat\rho}\log p\bigr)p\Bigr].
\end{align}
Substituting this identity into~\eqref{eq:FPE} yields
\[
  \partial_t p
  = -\partial_{\hat\rho}\Bigl[\Bigl(
      b - \frac{1}{2}\partial_{\hat\rho}\Sigma^2
        - \frac{1}{2}\Sigma^2\partial_{\hat\rho}\log p
    \Bigr)p\Bigr],
\]
which is exactly~\eqref{eq:continuity}.
\end{proof}

\subsection{Statement and proof}

\begin{proposition}[Probability-flow representation of the one-point marginals]
\label{prop:PFODE}
Under the assumptions of Lemma~\ref{lem:cont}, and provided that
$v_{\mathrm{PF}}$ is locally Lipschitz in $\hat\rho$ on
$(0,\rho_{\max})$ (see \Cref{prop:wellposed} for sufficient
conditions), the following holds.  Fix any $t_0>0$ at which the
one-point marginal $p(\cdot;x,t_0)$ is absolutely continuous with
$p>0$ on $(0,\rho_{\max})$.  Then the deterministic ODE
\begin{equation}\label{eq:PFODE}
  \frac{\dd\hat\rho}{\dd t}
  = \underbrace{b(\hat\rho,x,t)}_{\text{(I) conditional advection}}
  \;\underbrace{-\;\frac{1}{2}\partial_{\hat\rho}\Sigma^2(\hat\rho,x)}_{\text{(II) It\^{o} drift}}
  \;\underbrace{-\;\frac{1}{2}\Sigma^2(\hat\rho,x)
      \partial_{\hat\rho}\log p(\hat\rho;x,t)}_{\text{(III) score}},
\end{equation}
transports $p(\cdot;x,t_0)$ forward: if $\hat\rho(t_0)\sim
p(\cdot;x,t_0)$, then $\hat\rho(t)\sim p(\cdot;x,t)$ for all
$t\in[t_0,T]$.
With a computable closure $b_\phi$, the practical transportation-model
counterpart is
\begin{equation}\label{eq:PFODE-closed}
  \frac{\dd\hat\rho}{\dd t}
  = b_\phi(\hat\rho,x,t)
    - \frac{1}{2}\partial_{\hat\rho}\Sigma^2(\hat\rho,x)
    - \frac{1}{2}\Sigma^2(\hat\rho,x)
      \partial_{\hat\rho}\log p(\hat\rho;x,t).
\end{equation}
\end{proposition}

\begin{proof}
For $t\ge t_0$, the one-point law $p(\cdot;x,t)$ is absolutely
continuous and positive on $(0,\rho_{\max})$ by assumption, so
$v_{\mathrm{PF}}$ is well-defined.  Equation~\eqref{eq:continuity}
is the Liouville equation associated with the velocity
field~\eqref{eq:vPF}.  Under the Lipschitz regularity assumed above
(established in \Cref{prop:wellposed}(i)--(ii)), the characteristic
flow is well-defined and pushes forward $p(\cdot;x,t_0)$ to
$p(\cdot;x,t)$ for each $t\in[t_0,T]$.
Equation~\eqref{eq:PFODE-closed} is the computable transportation
surrogate obtained when the exact conditional drift $b$ is replaced by
a closure model $b_\phi$.
\end{proof}

\begin{remark}[Initial time]
The FPE initial condition~\eqref{eq:FPE-initial} is a delta mass,
which is measure-valued rather than absolutely continuous.  A
deterministic transport equation cannot spread a point mass into a
smooth density, so the PF-ODE representation applies strictly for
$t>0$, once the parabolic diffusion has already spread the law into a
smooth positive density.  In practice, one may equivalently formulate
the PF-ODE from a mollified initial law
$p_\epsilon(\hat\rho;x,0) = \delta_\epsilon(\hat\rho - \rho_0(x))$
and consider the limit $\epsilon\downarrow 0$; this is consistent with
the mollified initial condition used in the boundary-condition loss
$\mathcal{L}_{\text{BC}}$ of the training algorithm
(\Cref{sec:loss}).
\end{remark}

\subsection{Physical decomposition}

The three terms in~\eqref{eq:PFODE} have distinct transportation
interpretations.

\noindent\textbf{Term~(I): Conditional advection.}
The coefficient $b(\hat\rho,x,t)$ is the one-point imprint of the LWR
transport term.  It is not a free drift; it is the conditional
expectation of the pathwise advection $-f'(\rho)\partial_x\rho$ given
that the local density equals $\hat\rho$.

\noindent\textbf{Term~(II): Noise-induced drift (It\^{o} correction).}
This term is induced by state-dependent diffusion.  If the diffusion
intensity is strongest near capacity, it systematically redistributes
probability away from the most noise-sensitive density range.

\noindent\textbf{Term~(III): Score-driven correction.}
This term is the deterministic correction that converts stochastic
spreading into an equivalent probability-flow transport.  In the
learning architecture, it is the component represented by the score
network.

\begin{remark}[Contrast with PIDL]
In the PIDL framework of \citet{shi2021physics}, the PINN residual
enforces that a single trajectory satisfies a deterministic traffic-flow
model.  Here, the residual operates at the distributional level and
requires both a score representation and a transportation closure for
the unresolved advection term.
\end{remark}

\subsection{Well-posedness}

\begin{proposition}[Interior well-posedness and boundary compatibility]
\label{prop:wellposed}
Assume that $b\in C^1((0,\rho_{\max})\times Q_T)$,
$\Sigma^2\in C^2([0,\rho_{\max}]\times\mathcal{D})$, and
$p\in C^{2,1,1}((0,\rho_{\max})\times Q_T)$ with $p>0$ on the
interior. Then:
\begin{enumerate}[label=\textup{(\roman*)}]
  \item On every compact interval
    $I_\varepsilon=[\varepsilon,\rho_{\max}-\varepsilon]$, the score
    $s=\partial_{\hat\rho}\log p$ and the velocity
    $v_{\mathrm{PF}}$ are bounded and locally Lipschitz in $\hat\rho$.
  \item For every initial value $\hat\rho(0)\in I_\varepsilon$,
    ODE~\eqref{eq:PFODE} has a unique solution up to the first exit time
    from $I_\varepsilon$.
  \item If, in addition, the boundary compatibility conditions
    \begin{equation}\label{eq:boundary-compatibility}
      \liminf_{\hat\rho\downarrow 0} v_{\mathrm{PF}}(\hat\rho,x,t)\ge 0,
      \qquad
      \limsup_{\hat\rho\uparrow \rho_{\max}} v_{\mathrm{PF}}(\hat\rho,x,t)\le 0,
    \end{equation}
    hold for all $(x,t)\in Q_T$, then the interval
    $[0,\rho_{\max}]$ is forward invariant and the solution exists on
    $[0,T]$.
\end{enumerate}
\end{proposition}

\begin{proof}
Part~(i) follows from the assumed smoothness of $b$, $\Sigma^2$, and
$p$, together with positivity of $p$ on the interior.  Part~(ii) is a
direct application of the Picard--Lindel\"of theorem.
Part~(iii) follows from the inward-pointing boundary condition in
\eqref{eq:boundary-compatibility}.
\end{proof}

\begin{remark}[Scope of the well-posedness statement]\label{rem:scope-claim}
	The well-posedness statement established here is local to the interior
	state space, together with compatibility at the physical boundaries.
	A global result does not follow from $f'(0)>0$ and
	$f'(\rho_{\max})<0$ alone, since these conditions do not determine the
	sign of the conditional drift at the boundaries.
\end{remark}

\begin{remark}[Shock regime: analysis and practical considerations]
\label{rem:shock}
When Assumption~\ref{ass:smooth} fails and shocks form, the pointwise
conditional drift $b(\hat\rho,x,t)
= \EE[-f'(\rho)\partial_x\rho\mid\rho(x,t)=\hat\rho]$ involves
$\partial_x\rho$, which diverges at the shock location.  This is the
principal theoretical limitation of the smooth-regime theory.  We
discuss the issue along four dimensions.

\textit{(i) Methodological intuition for the score near shocks.}
Under the smooth-regime assumptions, the Fokker--Planck
equation~\eqref{eq:FPE} is parabolic in $\hat\rho$, and the score
$s = \partial_{\hat\rho}\log p$ is bounded on the interior.  When
shocks form and the conditional drift $b$ becomes singular, the
smooth-regime assumptions no longer hold, and we cannot formally
guarantee that the score remains bounded without additional
regularity analysis (e.g., under viscous regularisation).
However, the following heuristic is plausible: since diffusion acts in
the $\hat\rho$-direction and the singularity is in the $x$-direction
(through $\partial_x\rho$), the density-space smoothing from the
Brownian forcing may preserve score regularity even when the spatial
field is discontinuous.  Rigorous verification of this expectation is
deferred to future work.

\textit{(ii) Distributional structure near shocks.}
Near a shockwave, the density $\rho(x,t)$ transitions rapidly between
an upstream state $\rho_-$ and a downstream state $\rho_+$.  The
one-point distribution $p(\hat\rho;x,t)$ at a location traversed by
the shock is expected to become bimodal, with mass concentrated near
both $\rho_-$ and $\rho_+$.  This is precisely the kind of
non-Gaussian distributional structure that our framework is designed
to capture and that point-estimate methods cannot represent.

\textit{(iii) Learned closure as a regularised surrogate.}
In the proposed architecture (\Cref{sec:method}), the conditional
drift is not evaluated analytically but is represented by the learned
closure module $b_\phi(\hat\rho,x,t)$.  In practice, the neural
closure can act as a regularised surrogate in shock-affected regimes,
but this should be interpreted as an approximation heuristic rather
than a consequence of the present theory.  The physics loss
$\mathcal{L}_{\text{PF}}$ penalises the FPE residual at collocation
points; near shocks, the residual will typically be larger, and the
learned closure adapts subject to the network's smoothness prior.
This is analogous to how PINNs for deterministic conservation laws
often behave in practice: the network solution approximates a viscous
regularisation without explicitly introducing artificial viscosity.

\textit{(iv) Formal regularisation strategies.}
For applications requiring stronger theoretical guarantees, two
strategies are available:
(R1)~\emph{Viscous regularisation}: augmenting the SLWR dynamics with a
second-order viscous term, i.e., replacing $-\partial_x f(\rho)$ with
$-\partial_x f(\rho) + \epsilon\,\partial_{xx}\rho$ for small
$\epsilon > 0$, restores Assumption~\ref{ass:smooth} globally, and the
entire theory (FPE, PF-ODE, well-posedness) applies without
modification.  The resulting solutions converge to the entropy solution
as $\epsilon\downarrow 0$.
(R2)~\emph{Weak-form FPE}: interpreting the one-point forward equation
in the distributional sense, possibly using kinetic formulations in
the spirit of \citet{debussche2010scalar}, would extend the theory to
discontinuous density fields.  This is a substantial mathematical
undertaking that we defer to future work.
\end{remark}

\section{Physics-Informed Score Matching}\label{sec:method}

This section presents the complete estimation framework.  The
architecture mirrors the two-component structure of the PIDL framework
\citep{shi2021physics}, comprising an estimation network coupled with a
physics evaluator, but operates at the distributional level rather
than the trajectory level.

\subsection{Problem statement}

Sparse observations are collected from traffic sensors:
\begin{itemize}
  \item \textbf{Loop detectors} at fixed locations
    $\{x_\ell\}_{\ell=1}^{N_L}$ providing density
    $\rho^{\text{obs}}(x_\ell,t_i)$ and/or speed
    $u^{\text{obs}}(x_\ell,t_i)$ at discrete times.
  \item \textbf{Probe vehicles} traversing the segment along
    trajectories $\{(x_j(\tau),\tau)\}$, providing speed
    measurements $u^{\text{obs}}(x_j(\tau),\tau)$.
\end{itemize}
The goal is to produce a data-conditioned distributional estimate
$p_\theta(\hat\rho;x,t)$ for all $(x,t)\in Q_T$ that is
(i)~anchored to the available observations $\mathcal{O}$ via denoising
score matching, and (ii)~regularised by the one-point Fokker--Planck
equation~\eqref{eq:FPE}, which encodes the stochastic traffic physics
as a structural prior.

\begin{remark}[Nature of the distributional target]
\label{rem:stat-target}
The distribution $p(\hat\rho;x,t)$ derived in \Cref{sec:FPE} is the
\emph{forward stochastic law}: the one-point marginal of the density
field under the Brownian forcing across realisations.  In practice,
the estimation problem typically involves a \emph{single realisation}
observed through sparse sensors, and the object of operational interest
is the posterior $p(\hat\rho;x,t\mid\mathcal{O})$---the distribution
of the latent traffic state conditioned on one set of observations.
These two objects are distinct.  The present framework uses the forward
law's Fokker--Planck equation as a \emph{physics-informed structural
prior} that regularises a score network trained on the available
observations via DSM.  The resulting estimate $p_\theta(\hat\rho;x,t)$
is therefore best interpreted as an instance-specific approximation:
a distributional fit for the observed scenario whose structure is
constrained by, but not identical to, the forward stochastic law.
This is analogous to how a deterministic PINN produces a solution
consistent with observations and the PDE, without requiring the PDE
to be the full posterior model.  A fully Bayesian formulation---in
which $p(\hat\rho;x,t\mid\mathcal{O})$ is derived as a posterior
through filtering or variational inference with the stochastic LWR
as the prior process---would provide a more rigorous statistical
foundation but is beyond the scope of this paper.
\end{remark}

\subsection{Score network}\label{sec:score-net}

\begin{definition}[Score network]\label{def:score-net}
The score network is a neural network
$s_\theta:(0,\rho_{\max})\times\mathcal{D}\times[0,T]\to\RR$
that approximates the score function:
\begin{equation}\label{eq:score-approx}
  s_\theta(\hat\rho;x,t) \;\approx\;
  \partial_{\hat\rho}\log p(\hat\rho;x,t).
\end{equation}
\end{definition}

The input is a triplet $(\hat\rho,x,t)$ where
$\hat\rho\in(0,\rho_{\max})$ is a query point in density space and
$(x,t)$ specifies the spatio-temporal location.  The output is a
scalar $s_\theta\in\RR$.  The architecture is a fully-connected
feedforward network with $L$ hidden layers of width $H$ and $\tanh$
activation.  Sinusoidal positional encoding is applied to $(x,t)$:
\begin{equation}\label{eq:pos-enc}
  \gamma(z) = \bigl(\sin(2^0\pi z),\cos(2^0\pi z),\dots,
  \sin(2^{M-1}\pi z),\cos(2^{M-1}\pi z)\bigr),
\end{equation}
yielding an effective input dimension of $1+2\times 2M$.

\begin{remark}[Scalar score]\label{rem:scalar-score}
Since the stochastic forcing in~\eqref{eq:SLWR} produces diffusion
only in the $\hat\rho$-direction, the Fokker--Planck equation is
parabolic in $\hat\rho$ alone and the score is a scalar function.
This substantially reduces the required network capacity compared to
multi-dimensional score networks used in image generation
\citep{song2021score}.
\end{remark}

\subsection{Physics residual evaluator}\label{sec:phys-eval}

The exact one-point theory in \Cref{sec:FPE,sec:PFODE} shows that the
physics residual cannot be built from the score network alone: one must
also specify a transportation closure for the unresolved advection term.
Accordingly, we introduce an auxiliary advection-closure module
$b_\phi(\hat\rho,x,t)$, or equivalently a structured closure of the
form $b_\phi(\hat\rho,x,t)=-f'(\hat\rho)m_\phi(\hat\rho,x,t)$.
Define the closed probability-flow velocity
\begin{equation}\label{eq:vPF-theta-phi}
  v_{\theta,\phi}(\hat\rho,x,t)
  := b_\phi(\hat\rho,x,t)
     - \frac{1}{2}\partial_{\hat\rho}\Sigma^2(\hat\rho,x)
     - \frac{1}{2}\Sigma^2(\hat\rho,x)\,s_\theta(\hat\rho;x,t).
\end{equation}
From the continuity equation
$\partial_t p + \partial_{\hat\rho}(v_{\theta,\phi}p)=0$, one has
\begin{equation}\label{eq:logp-evolution}
  \partial_t \log p = -\partial_{\hat\rho}v_{\theta,\phi}
  - v_{\theta,\phi}\,\partial_{\hat\rho}\log p.
\end{equation}
Differentiating with respect to $\hat\rho$ and substituting
$s_\theta\approx\partial_{\hat\rho}\log p$ yields the score-form
forward-equation residual
\begin{equation}\label{eq:score-FPE-residual}
  \mathcal{R}_{\theta,\phi}(\hat\rho,x,t)
  := \partial_t s_\theta
   + v_{\theta,\phi}\,\partial_{\hat\rho}s_\theta
   + (\partial_{\hat\rho}v_{\theta,\phi})\,s_\theta
   + \partial_{\hat\rho}^2 v_{\theta,\phi}.
\end{equation}
All terms are computable by automatic differentiation.  This is the
distributional analogue of the deterministic PDE residual in the PIDL
framework, but now framed as a transportation closure plus score
representation.

\subsection{Loss function}\label{sec:loss}

The joint training loss has three components:
\begin{equation}\label{eq:loss}
  \boxed{
  \mathcal{L}(\theta,\phi) = \mathcal{L}_{\text{SM}}(\theta)
  + \lambda\,\mathcal{L}_{\text{PF}}(\theta,\phi)
  + \lambda_{\text{BC}}\,\mathcal{L}_{\text{BC}}(\theta,\phi)
  }
\end{equation}

\subsubsection{Score matching loss $\mathcal{L}_{\text{SM}}$}

We employ denoising score matching \citep[DSM;][]{vincent2011connection}, adapted to
the traffic setting.  At sensor locations where $\rho^{\text{obs}}_i$
is observed, we perturb each observation with known noise and match
the score of the perturbed distribution.  For each observation at
$(x_i,t_i)$:
\begin{enumerate}
  \item Draw $\epsilon\sim\mathcal{N}(0,1)$ and form
    $\tilde\rho_i = \rho^{\text{obs}}_i + \sigma_{\text{DSM}}\epsilon$.
  \item The score of the Gaussian perturbation kernel
    $q_\sigma(\tilde\rho\mid\rho^{\text{obs}})
    = \mathcal{N}(\tilde\rho;\rho^{\text{obs}},\sigma_{\text{DSM}}^2)$
    is $-\epsilon/\sigma_{\text{DSM}}$.
\end{enumerate}
The DSM loss is:
\begin{equation}\label{eq:DSM}
  \mathcal{L}_{\text{SM}}(\theta)
  = \frac{1}{N_{\text{obs}}}\sum_{i=1}^{N_{\text{obs}}}
  \EE_{\epsilon\sim\mathcal{N}(0,1)}\Bigl[
    \bigl|s_\theta(\tilde\rho_i;x_i,t_i)
    + \epsilon/\sigma_{\text{DSM}}\bigr|^2\Bigr].
\end{equation}
Following \citet{song2019generative}, we use multiple noise scales
$\{\sigma^{(l)}_{\text{DSM}}\}_{l=1}^{N_\sigma}$ in geometric
progression to capture both global structure and local precision,
with annealing weights $w_l\propto(\sigma^{(l)})^2$.
Since the density domain $(0,\rho_{\max})$ is bounded, Gaussian
perturbations at larger noise scales can produce samples $\tilde\rho_i$
outside the physical range.  The DSM target
$-\epsilon/\sigma_{\text{DSM}}$ in~\eqref{eq:DSM} is exact only for
the unconstrained Gaussian perturbation kernel.  On the bounded domain,
clipping, reflection, or truncation changes the perturbation law and
therefore changes the target score.  In practical implementations,
a boundary-aware perturbation kernel with its correspondingly adjusted
score target should be used when boundary effects are non-negligible;
at the smallest noise scales in a multi-scale DSM schedule, boundary
effects are typically negligible provided the observations lie well
within $(0,\rho_{\max})$.

When only speed $u^{\text{obs}}$ is available (e.g., from probe
vehicles), we exploit the fundamental diagram $u=v(\rho)$ to convert
to density: $\rho^{\text{obs}}=v^{-1}(u^{\text{obs}})$.

\begin{remark}[Regularisation of the score from sparse observations]
\label{rem:DSM-identifiability}
In standard generative modelling, DSM learns a score from independent
samples drawn from the target distribution.  In the traffic setting,
sensor observations at a given $(x,t)$ typically come from a single
realisation of the traffic field, not from repeated independent draws
of $p(\hat\rho;x,t)$.  The inverse problem is nevertheless regularised
by two complementary information sources: (i)~the DSM loss anchors the
score at observed sensor locations across different $(x_i,t_i)$,
effectively pooling information across space and time, and (ii)~the
physics loss $\mathcal{L}_{\text{PF}}$ enforces the Fokker--Planck
equation at collocation points throughout the $(\hat\rho,x,t)$ domain,
propagating distributional constraints from observed to unobserved
regions via the conservation-law dynamics.  Multi-scale DSM with noise
annealing further regularises the score at scales not directly resolved
by the sparse observations.  We note that this regularisation argument
is methodological rather than a formal identifiability theorem; in
particular, the interplay between the learned score $s_\theta$ and the
learned closure $b_\phi$ introduces flexibility that may require
additional structural constraints in practice.

To mitigate this flexibility, the structured closure
$b_\phi(\hat\rho,x,t) = -f'(\hat\rho)\,m_\phi(\hat\rho,x,t)$
provides a significant reduction in degrees of freedom: the known
characteristic-speed factor $f'(\hat\rho)$ is imposed analytically,
and only the conditional-gradient surrogate $m_\phi$ is learned.  This
factorisation ensures that $b_\phi$ has the correct sign structure
(advecting in the characteristic direction) and that in the
zero-noise limit ($\sigma_k\to 0$), the closure is designed to recover
the deterministic LWR dynamics, since the intended target
$m_\phi\to\partial_x\rho$ reduces the equation to the classical
LWR form.  The
narrower mean-field variant
$b_\phi = -f'(\hat\rho)\,\partial_x\bar\rho_\phi(x,t)$ constrains
the closure further to depend on a single auxiliary mean-state field,
reducing the learnable component to a function of $(x,t)$ alone and
eliminating the $\hat\rho$-dependence of the gradient surrogate
entirely.  These progressively restrictive structures provide a
hierarchy of inductive biases that trade off expressiveness against
identifiability; the appropriate level can be selected based on data
availability and validated via the ablation studies.
\end{remark}

\subsubsection{Physics loss $\mathcal{L}_{\text{PF}}$}

A set of collocation points
$\mathcal{C}=\{(\hat\rho_j,x_j,t_j)\}_{j=1}^{N_C}$ is sampled from
$(0,\rho_{\max})\times\mathcal{D}\times[0,T]$ via Latin hypercube
sampling.  The physics loss is:
\begin{equation}\label{eq:phys-loss}
  \mathcal{L}_{\text{PF}}(\theta,\phi)
  = \frac{1}{N_C}\sum_{j=1}^{N_C}
  \bigl|\mathcal{R}_{\theta,\phi}(\hat\rho_j,x_j,t_j)\bigr|^2.
\end{equation}
A key structural difference from standard PINN implementations is that
collocation occurs in the three-dimensional domain
$(\hat\rho,x,t)$ rather than in $(x,t)$ alone, reflecting the
additional density-space dimension of the Fokker--Planck equation.

\subsubsection{Boundary condition loss $\mathcal{L}_{\text{BC}}$}

This term enforces the zero-flux boundary condition
$J(\hat\rho;x,t)=0$ at $\hat\rho\in\{0,\rho_{\max}\}$ together with the
initial condition $p(\hat\rho;x,0)=\delta_\epsilon(\hat\rho-\rho_0(x))$, where
$\delta_\epsilon$ is a mollified delta function.  The density
$p_\theta$ is reconstructed from the score via numerical quadrature:
$p_\theta(\hat\rho;x,t)\propto
\exp\bigl(\int_{\rho_*}^{\hat\rho}s_\theta(\rho';x,t)\,\dd\rho'\bigr)$,
where $\rho_*\in(0,\rho_{\max})$ is an arbitrary reference density. The
corresponding boundary-flux surrogate is
$J_{\theta,\phi}=b_\phi p_\theta-\tfrac12\partial_{\hat\rho}(\Sigma^2 p_\theta)$.

\subsection{Training algorithm}\label{sec:training}

\begin{algorithm}[ht]
\caption{Joint Score Matching + Physics-Informed Training}
\label{alg:training}
\begin{algorithmic}[1]
\Require Observations $\mathcal{O}$; model parameters
  $(f,\sigma_k,e_k,K)$; hyperparameters
  $(\lambda,\lambda_{\text{BC}},\sigma_{\text{DSM}},L,H,M)$;
  learning rate $\eta$; batch sizes $(B_{\text{obs}},B_C,B_B)$.
\Ensure Trained score network $s_{\theta^*}$ and closure module $b_{\phi^*}$.
\State Initialise $\theta$ and $\phi$ (Xavier initialisation).
\For{epoch $=1,2,\ldots,N_{\text{warm}}$}
  \State Sample $B_{\text{obs}}$ observations; draw
    $\epsilon\sim\mathcal{N}(0,I)$; compute $\mathcal{L}_{\text{SM}}$.
  \State Sample $B_C$ collocation points via Latin hypercube in
    $(0,\rho_{\max})\times\mathcal{D}\times[0,T]$.
  \State Evaluate $s_\theta$, $b_\phi$, and the required derivatives
    $\partial_{\hat\rho}s_\theta$, $\partial_{\hat\rho}^2 s_\theta$,
    $\partial_t s_\theta$,
    $\partial_{\hat\rho}b_\phi$, $\partial_{\hat\rho}^2 b_\phi$
    via automatic differentiation.
  \State Compute $\mathcal{R}_{\theta,\phi}$ at each collocation
    point; compute $\mathcal{L}_{\text{PF}}$.
  \State Sample $B_B$ boundary points; compute $\mathcal{L}_{\text{BC}}$.
  \State $\mathcal{L}\gets\mathcal{L}_{\text{SM}}
    +\lambda\mathcal{L}_{\text{PF}}
    +\lambda_{\text{BC}}\mathcal{L}_{\text{BC}}$.
  \State Update $(\theta,\phi)\gets(\theta,\phi)-\eta\nabla_{(\theta,\phi)}\mathcal{L}$
    (Adam optimiser with cosine learning rate decay).
\EndFor
\State Fine-tune via L-BFGS until $|\Delta\mathcal{L}|<10^{-6}$.
\end{algorithmic}
\end{algorithm}

The hyperparameter $\lambda$ is balanced adaptively using
neural tangent kernel (NTK) rescaling \citep{wang2022ntk}, ensuring
that the gradient magnitudes of $\mathcal{L}_{\text{SM}}$ and
$\lambda\mathcal{L}_{\text{PF}}$ remain of comparable order.
When the noise structure $(\sigma_k,e_k)$ is unknown, these
parameters are encoded as trainable variables in the physics evaluator
and optimised jointly with $(\theta,\phi)$.  In this case, the
parameterisation of $\sigma_k$ must be constrained to enforce the
endpoint vanishing conditions
$\sigma_k(0)=\sigma_k(\rho_{\max})=0$ and the interior positivity
required by Assumption~\ref{ass:noise}; for example, one may write
$\sigma_k(\rho) = \rho(\rho_{\max}-\rho)\,\tilde\sigma_k(\rho)$
with $\tilde\sigma_k$ unconstrained.

\subsection{Inference}\label{sec:inference}

Once $s_{\theta^*}$ (and, when used, $b_{\phi^*}$) is trained, the
distributional estimate at any query point $(x^*,t^*)$ is recovered via:
\begin{equation}\label{eq:log-density-recovery}
  \log p_{\theta^*}(\hat\rho;x^*,t^*)
  = \int_{\rho_*}^{\hat\rho}s_{\theta^*}(\rho';x^*,t^*)\,\dd\rho'
  + C(x^*,t^*),
\end{equation}
where $\rho_*\in(0,\rho_{\max})$ is an arbitrary reference density and
$C(x^*,t^*)$ is the log-normalisation constant determined by the unit
total mass requirement
$\int_0^{\rho_{\max}} p_{\theta^*}(\hat\rho;x^*,t^*)\,\dd\hat\rho = 1$,
evaluated by Gauss--Legendre quadrature with $N_Q=100$ nodes on
$[0,\rho_{\max}]$.
From $p_{\theta^*}$, the following summary statistics are computed:
\begin{align}
  \text{Conditional mean:}\;&\;
  \bar\rho(x^*,t^*) = \textstyle\int_0^{\rho_{\max}}\hat\rho\,
    p_{\theta^*}\,\dd\hat\rho, \label{eq:cond-mean} \\
  \text{Conditional std:}\;&\;
  \sigma_\rho = \bigl[\textstyle\int(\hat\rho-\bar\rho)^2
    p_{\theta^*}\,\dd\hat\rho\bigr]^{1/2}, \label{eq:cond-std} \\
  \text{95\% credible interval:}\;&\;
  [\rho_{0.025},\;\rho_{0.975}]. \label{eq:CI}
\end{align}
Quadrature-based density recovery is the primary inference route.
In principle, the closed PF-ODE~\eqref{eq:PFODE-closed} could also be
used for sample generation by evolving an ensemble of initial
conditions drawn from a mollified initial law; the design of such a
sampling procedure is left to the numerical implementation.

\subsection{Structural comparison with PIDL}\label{sec:PIDL-compare}

\Cref{fig:arch} illustrates the architectural correspondence between
the PIDL framework of \citet{shi2021physics} and the proposed method.

\begin{figure}[ht]
\centering
\begin{tikzpicture}[
  block/.style={draw, rounded corners, minimum height=0.9cm,
    minimum width=2.2cm, align=center, font=\small},
  arrow/.style={-{Stealth[length=2.5mm]}, thick},
]
\node[block,fill=blue!10] (punn) at (0,3.5)
  {PUNN\\$(x,t)\!\to\!(\hat\rho,\hat u)$};
\node[block,fill=red!10] (pinn) at (5.5,3.5)
  {PINN\\ARZ residual};
\node[block,fill=yellow!15] (out1) at (9.5,3.5)
  {Point\\estimate};
\draw[arrow] (punn)--(pinn); \draw[arrow] (pinn)--(out1);
\node[font=\footnotesize\itshape] at (2.75,4.3) {PIDL (Shi et al., 2021)};
\node[block,fill=blue!10] (score) at (0,1)
  {Score Net\\$(\hat\rho,x,t)\!\to\!s_\theta$};
\node[block,fill=red!10] (phys) at (5.5,1)
  {FPE Residual\\$\mathcal{R}_{\theta,\phi}$};
\node[block,fill=yellow!15] (out2) at (9.5,1)
  {$p_\theta(\hat\rho;x,t)$};
\draw[arrow] (score)--(phys); \draw[arrow] (phys)--(out2);
\node[font=\footnotesize\itshape] at (2.75,1.8) {This paper};
\draw[-{Stealth[length=3mm]},ultra thick,dashed,gray]
  (2.75,2.8)--(2.75,2.2);
\end{tikzpicture}
\caption{Architectural comparison between PIDL and the proposed
framework.  Both share a two-component structure (estimation network +
physics evaluator), but the proposed framework operates at the
distributional level: the network learns a score function, the physics
evaluator regularises via the Fokker--Planck equation, and the output is a
data-conditioned distributional estimate.}
\label{fig:arch}
\end{figure}
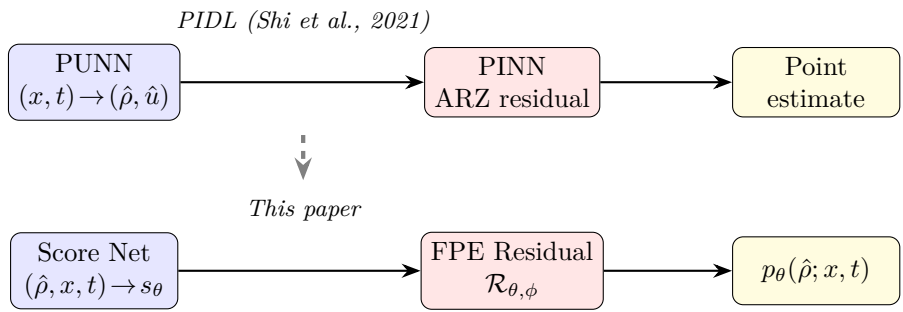

\Cref{tab:component-compare} provides a component-level comparison.

\begin{table}[ht]
\centering
\caption{Component-level comparison between PIDL and the proposed
framework.}
\label{tab:component-compare}
\renewcommand{\arraystretch}{1.25}
\small
\begin{tabular}{l|cc}
\toprule
\textbf{Component} & \textbf{PIDL} & \textbf{This paper} \\
\midrule
Network input & $(x,t)$ & $(\hat\rho,x,t)$ \\
Network output & $(\hat\rho,\hat u)$ point values &
  $s_\theta$ (score, scalar) \\
Physics evaluator & ARZ residual $(\hat f_1,\hat f_2)$ &
  FPE residual $\mathcal{R}_{\theta,\phi}$ \\
Data loss & $\sum|\hat\rho-\rho^{\text{obs}}|^2$ &
  DSM: $\sum|s_\theta+\epsilon/\sigma|^2$ \\
Physics loss & $\sum|\hat f_1|^2+|\hat f_2|^2$ &
  $\sum|\mathcal{R}_{\theta,\phi}|^2$ \\
Collocation domain & $(x,t)$, 2D & $(\hat\rho,x,t)$, 3D \\
Output & Single $\hat\rho(x,t)$ &
  Full $p_\theta(\hat\rho;x,t)$ \\
UQ & None (built-in) & Credible intervals from $p$ \\
\bottomrule
\end{tabular}
\end{table}

\section{Discussion and Future Work}\label{sec:discussion}

We derive the exact one-point Fokker--Planck equation for the Ito-type stochastic LWR model and propose a physics-informed score-based learning architecture for the resulting problem. This section discusses the practical significance of the framework, its computational characteristics, current limitations, possible extensions to stochastic fundamental diagrams at the link and network levels, and the planned empirical validation.

\subsection{Practical implications}

The distributional output enables risk-aware traffic
management.  The probability that density exceeds a critical threshold
$\rho_c$ (e.g., the onset of congestion) can be computed as
$\int_{\rho_c}^{\rho_{\max}}p_\theta(\hat\rho;x,t)\,\dd\hat\rho$.
This quantity directly informs variable speed limit decisions and
congestion pricing with quantified confidence levels, which is not
achievable with point-estimate methods.

\subsection{Computational considerations}

The proposed architecture adds computational overhead relative to
standard PIDL due to the three-dimensional collocation in
$(\hat\rho,x,t)$ and multi-scale DSM.  However, inference is
efficient: evaluating $p_\theta$ at a single $(x,t)$ requires $N_Q$
forward passes (${\sim}100$), which is trivially parallelisable.
Compared with a Monte Carlo forward solver, the proposed method solves
an \emph{inverse} estimation problem from sparse data, whereas Monte Carlo addresses the forward
simulation problem.  Concrete runtime comparisons will be reported
with the numerical experiments.

\subsection{Extension to the stochastic fundamental diagram}

An important next step is to extend the proposed framework to the
stochastic fundamental diagram at both the link and network levels.

\paragraph{Link-level stochastic fundamental diagram.}
At a single link or detector location $x$, the classical fundamental
diagram $q = f(\rho)$ is a deterministic curve.  Under the It\^{o}-type
SLWR, the density at location $x$ is a random variable with
distribution $p(\hat\rho;x,t)$ governed by the Fokker--Planck equation.
The corresponding flow distribution is obtained via the change of
variables formula.  Since standard concave fundamental diagrams are not
globally invertible on $[0,\rho_{\max}]$ (sub-capacity flows correspond
to two density branches), the flow density is computed as a sum over
preimages:
$p(\hat{q};x,t) = \sum_{j:\,f(\rho_j)=\hat{q}}
p(\rho_j;x,t)/|f'(\rho_j)|$,
or more generally through any functional
$\EE[g(f(\rho))\mid\rho(x,t)=\hat\rho]$.
The scatter observed in empirical flow-density measurements at a fixed
detector is then a direct manifestation of the stochastic forcing: the
width and shape of $p(\hat{q};x,t)$ at a given mean density are
determined by the noise parameters $\sigma_k$ and $e_k(x)$.  This
provides a \emph{link-level stochastic fundamental diagram} as a
naturally derived quantity.  The joint distribution of $(\rho, q)$ at a fixed location,
or equivalently the marginal flow distribution $p(\hat{q};x,t)$
obtained by pushforward, characterises day-to-day flow variability
and is a quantity of direct operational value for reliability-based
traffic management.  (Note that under a deterministic flux relation
$q=f(\rho)$, conditioning on exact $\hat\rho$ makes $q$ deterministic;
the meaningful stochastic objects are the marginal and joint
distributions, or bin-conditioned empirical relations.)

\paragraph{Network-level stochastic macroscopic fundamental diagram.}
While the concept of a macroscopic flow--density relationship has long
been known in traffic theory, the network-level MFD, as characterised
empirically by \citet{geroliminis2008existence} and others, exhibits a
well-defined but \emph{scattered} relationship between the
network-average flow $\bar{q}$ and network-average density $\bar\rho$.
This scatter has been attributed primarily to the spatial heterogeneity
of the density distribution within the region
\citep{mazloumian2010spatial, leclercq2014macroscopic,
ambuehl2023understanding}.  The noise structure
$\sum_k \sigma_k(\rho)\,e_k(x)\,\dd W^k_t$ generates precisely the
kind of spatially heterogeneous density perturbations that have been
identified as the source of network MFD scatter.

A natural concern is whether the shock-regime limitation of the
pointwise FPE (\Cref{rem:shock}) undermines the network-level
extension.  While spatial averaging improves regularity (the MFD
variables
$\bar\rho(t) = \frac{1}{|\mathcal{A}|}\int_\mathcal{A}\rho(x,t)\,\dd x$
and
$\bar{q}(t) = \frac{1}{|\mathcal{A}|}\int_\mathcal{A}f(\rho(x,t))\,\dd x$
are continuous in time even when the density field contains shocks),
deriving a closed, low-dimensional stochastic evolution for
$(\bar\rho,\bar{q})$ raises its own closure challenges that are
distinct from the pointwise case.  We therefore present the following
research avenues as conjectural extension paths rather than
consequences secured by the present theory:
(i)~derivation of a network-level forward equation for
$p(\bar{q},\bar\rho;t)$ via spatial aggregation of the link-level
stochastic dynamics;
(ii)~estimation of the distribution
$p(\bar{q}\mid\bar\rho;t)$ via the probability flow ODE
framework; and (iii)~physics-informed score matching at the network
level with an aggregated forward equation as the physics constraint.

These extensions would connect the stochastic traffic flow modelling
tradition to the fundamental diagram literature at both scales: the
link-level stochastic FD follows from the present framework, while
the network-level stochastic MFD remains an open problem for which
the current work provides a plausible mathematical starting point.

\subsection{Limitations}

Several limitations should be acknowledged.  First, the well-posedness
results (\Cref{prop:wellposed}) require the smooth-solution regime
(Assumption~\ref{ass:smooth}); extending the theory to entropy
solutions in the presence of shocks remains an open mathematical
problem.  Second, the one-point marginal FPE does not capture spatial
correlations in the density field; extension to multi-point marginals
or richer spatial closure models would provide more complete
distributional information.  Third, the current framework addresses
the first-order LWR model; extension to second-order models (e.g.,
ARZ) would enable joint density--velocity distributional estimation.
Fourth, online adaptive retraining as new sensor data arrive would
further enhance practical applicability.

\paragraph{Future empirical validation.}
Empirical validation is deferred to a subsequent revision.  We plan
to validate the framework first on synthetic data generated from Monte
Carlo realisations of the stochastic LWR model, where the inferred
one-point density distributions can be compared directly with the
reference stochastic solution.  We will then compare the method on
real traffic data against deterministic PIDL baselines to evaluate the
practical value of distributional, rather than purely pointwise,
traffic state estimation.  Performance will be assessed in terms of
point accuracy, calibration, and distributional quality.

\section{Conclusion}\label{sec:conclusion}

This paper develops a theoretical foundation for a
physics-informed deep learning framework for distributional traffic
state estimation under stochastic LWR dynamics.  Its main contributions include:

First, it formulates an It\^{o}-type stochastic LWR model with finite-dimensional
Brownian noise and density-dependent diffusion,
extending random-parameter stochastic LWR models from static
parametric uncertainty to dynamic stochastic forcing over the estimation horizon.

Second, it derives the Fokker--Planck equation for the one-point
marginal density as an exact forward equation with an explicit
conditional drift term. It also derives the corresponding probability
flow ODE, yielding a deterministic transport representation that can
serve as a distributional physics constraint in a neural network once
a closure is specified.

Third, it proposes a distributional estimation framework: a score
network paired with an auxiliary advection-closure module, trained via
denoising score matching and a Fokker--Planck residual loss.  The
architecture targets the one-point distributional estimate of traffic
density, from which point estimates, credible intervals, and
congestion-risk metrics can be derived.

At the link level, the one-point distribution $p(\hat\rho;x,t)$
yields a stochastic fundamental diagram in which the scatter in
empirical flow-density measurements is a direct consequence of the
Brownian forcing.  At the network level, the framework suggests a
conjectural pathway toward a stochastic macroscopic fundamental
diagram, though deriving a closed network-level forward equation
raises additional closure challenges that remain open.  More broadly,
the idea of embedding stochastic conservation law physics at the
distributional level into deep learning architectures may find
application beyond traffic, in any domain where hyperbolic PDEs
govern uncertain dynamics observed through sparse measurements.
The main point is not simply that score-based learning is applied to
traffic estimation, but that the generative structure is derived from
stochastic traffic-flow physics, with closure and learning components
introduced only to obtain a trainable method. More broadly, the paper
shows how stochastic traffic-flow theory can be cast as a trainable
distributional estimation framework, providing a physics-grounded
alternative to purely pointwise traffic-state estimation and to
generic diffusion models that ignore conservation-law structure.

\bibliographystyle{elsarticle-harv}

\end{document}